\documentclass[11pt]{article}

\usepackage[a4paper,margin=25mm]{geometry}
\usepackage{amsmath,amssymb,amsthm}
\usepackage{physics}
\usepackage{bm}

\usepackage{graphicx}
\usepackage{hyperref}
\usepackage{cite}
\usepackage{arydshln}

\usepackage{pgfplots}
\pgfplotsset{compat=1.18}

\theoremstyle{plain}
\newtheorem{theorem}{Theorem}[section]
\newtheorem{proposition}[theorem]{Proposition}
\newtheorem{lemma}[theorem]{Lemma}
\newtheorem{corollary}[theorem]{Corollary}

\theoremstyle{definition}
\newtheorem{definition}[theorem]{Definition}

\theoremstyle{remark}
\newtheorem{remark}[theorem]{Remark}

\counterwithin{equation}{section}

\newcommand{\startappendix}{
  \setcounter{section}{1}
  \renewcommand{\thesection}{A}
  \setcounter{subsection}{0}
  \renewcommand{\thesubsection}{A.\arabic{subsection}}
  \setcounter{theorem}{0}
  \renewcommand{\thetheorem}{A.\arabic{theorem}}
  \setcounter{equation}{0}
  \renewcommand{\theequation}{A.\arabic{equation}}
}

\title{%
Standard Quantum Phase Estimation Detects All Eigenvalues via Randomized Initial States
}

\author{%
Yuki Izumi\thanks{yuki.izumi@ctc-g.co.jp}
\hspace{2em}
Hitoshi Kawahara\thanks{hitoshi.kawahara@ctc-g.co.jp}
\\[1em]
\normalsize ITOCHU Techno-Solutions Corporation, Tokyo, Japan
}

\date{}

\begin{document}
\maketitle

\begin{abstract}
Standard quantum phase estimation (QPE) has often been regarded as unsuitable for simultaneous detection of all eigenvalues, because it requires initial states with sufficient overlap with the target eigenstates.
In this paper, we show that this limitation is not inherent to the QPE circuit itself.
The output distribution of standard QPE can be written as a superposition of Fej\'er kernels weighted by the squared overlaps with the eigenmodes.
We prove that, if the initial state is independently drawn at each shot from a 1-design (in particular, by random selection of computational basis states), these mode weights are equalized in expectation, yielding a state-averaged QPE distribution that exhibits peaks at every eigenphase location.
In this sense, all eigenvalues become accessible without any modification of the standard QPE circuit; repeated eigenvalues appear through the aggregated weight of their eigenspaces.
For distinct eigenphases satisfying a separation condition, we further establish a rigorous peak-detection theory and derive a sufficient shot-count estimate for detecting all peaks.
We validate the theory through numerical experiments on a finite element method (FEM) matrix with $1{,}008$ degrees of freedom arising from computer-aided engineering (CAE).
\end{abstract}

\section{Introduction}\label{sec:introduction}

Quantum phase estimation (QPE) is one of the most important algorithms in quantum computing, serving as a foundation for integer factorization~\cite{shor1994algorithms,nielsen2010quantum}, quantum linear system solvers~\cite{harrow2009quantum}, and many other applications.
QPE is an algorithm for estimating the eigenphases of a unitary operator $U$ and serves as a core subroutine for a large family of quantum algorithms~\cite{martyn2021grand}.

Improving the accuracy and efficiency of QPE remains an active area of research.
One of the intrinsic limitations of QPE is spectral leakage: when an eigenphase cannot be exactly represented as a finite-length binary fraction on the measurement grid, the measurement probability leaks to bins other than the closest one~\cite{lim2024curvefitted}.
Under the constraints on the number of ancilla qubits in the NISQ setting, this leakage is particularly severe, and advanced classical post-processing methods such as Bayesian estimation~\cite{obrien2019quantum} and curve fitting~\cite{lim2024curvefitted} have been investigated to mitigate spectral leakage.
In the FTQC setting, a sufficient number of ancilla qubits can be allocated to suppress spectral leakage to a negligible level.

However, a more fundamental problem persists in the execution of QPE.
When standard QPE is applied to an initial state $|\psi\rangle = \sum_k c_k |\psi_k\rangle$ (a superposition of eigenstates $|\psi_k\rangle$ of $U$), the output distribution takes the form
\begin{align*}
    p_j = \sum_k |c_k|^2 \, \tilde{F}_N(2\pi(\theta_k - a_j)),
\end{align*}
where $\tilde{F}_N$ is the normalized Fej\'{e}r kernel, $\theta_k$ are the eigenphases, and $a_j = j/N$ are the measurement grid points.
Since the height of the peak corresponding to each eigenvalue is proportional to the weight $|c_k|^2$, peaks for modes with small $|c_k|^2$ are buried under spectral leakage from other modes, making detection difficult.
This has made the preparation of a ``good'' initial state---one with sufficient overlap to the target eigenstate---a de facto prerequisite for running QPE.

Considerable effort has been devoted to satisfying this prerequisite.
In quantum chemistry, the Hartree--Fock state provides a good approximation to the ground state, enabling high-precision estimation of a single eigenvalue (the ground state energy) when used as the initial state.
However, this approach relies on physical intuition specific to quantum chemistry and does not guarantee the general-purpose applicability of QPE.
Studies on QPE for multiple eigenvalues~\cite{obrien2019quantum} also operate under fixed initial states and perform estimation under nonuniform weights $|c_k|^2$, without addressing the bias itself.
In settings such as the generalized eigenvalue problem $\boldsymbol{\mathrm{K}}\boldsymbol{\mathrm{v}} = \lambda \boldsymbol{\mathrm{M}}\boldsymbol{\mathrm{v}}$ arising from the finite element method (FEM) in computer-aided engineering (CAE), where detection of all eigenvalues within a specified frequency band is required, preparing an initial state with sufficient overlap to the eigenmodes in the target band is inherently difficult---and this difficulty becomes more severe precisely as the problem size grows to the scale where QPE is needed.

At the root of this difficulty lies the fact that eigenvalues are intrinsic properties of the matrix.
Classical eigenvalue solvers take a matrix as input and return eigenvalues without requiring any additional prior information.
To make QPE a general-purpose eigenvalue detection tool, the dependence on the initial state must be eliminated entirely.

In this paper, we present a probabilistic approach to this problem.
We show that simply selecting the initial state randomly from the computational basis at each shot---which corresponds to a state 1-design construction---equalizes the weights $|c_k|^2$ to $1/M$ in expectation, where $M$ is the dimension of the matrix.
This equalization renders the QPE output distribution independent of the initial state, producing Fej\'{e}r kernel peaks of uniform height at every eigenvalue position.
We call this distribution the state-averaged QPE measure.

Furthermore, we develop a rigorous theory of peak detection under the state-averaged QPE measure.
Specifically, we define a detection set $C$ consisting only of grid points near eigenphase peaks, and prove the existence of a positive gap between the measurement probabilities of $a_j \in C$ and $a_j \notin C$ (Theorem~\ref{theorem:peak detection}).
Based on this gap, we provide an estimate of the shot count sufficient for detecting all eigenvalues via threshold testing on the empirical distribution (Theorem~\ref{theorem:shot estimate}).

In addition, the state-averaged formulation is also useful for obtaining rough prior information on eigenvalue locations.
Because it does not depend on a particular initial-state overlap pattern, the present approach can indicate the approximate locations of previously unknown eigenvalues without relying on prior guesses about where they lie.

The main contributions of this paper are as follows.
First, we show that standard QPE, without any circuit modification, makes all eigenvalues accessible using only 1-design randomized initial states.
Second, under a distinct-eigenphase separation assumption, we develop a peak detection theory based on properties of the Fej\'er kernel, provide an estimate of the shot count sufficient for detection of all distinct eigenvalue peaks, and establish a rigorous sample-based peak detection result for the detected set $\hat{C}$.
Third, we validate the theoretical claims through numerical experiments on an FEM matrix with $1{,}008$ degrees of freedom arising from CAE.

The remainder of this paper is organized as follows.
Section~\ref{sec:related work} discusses the relation of the present work to DOS estimation, multi-eigenvalue phase estimation, and modified QPE methods.
Section~\ref{sec:preliminaries} introduces the basic notions of QPE, the Fej\'{e}r kernel, block encoding, and 1-designs.
Section~\ref{sec:state averaged QPE} defines the state-averaged QPE measure and shows its realization via 1-designs.
Section~\ref{sec:peak detection} presents the peak detection theory and the shot count estimate.
Section~\ref{sec:numerical experiments} validates the theoretical claims through numerical experiments on an FEM matrix arising from CAE.
Section~\ref{sec:conclusion} summarizes the results and discusses future directions.
Proofs and technical details are deferred to the Appendix.

\section{Related Work}\label{sec:related work}

To clarify the position of the present work, we compare it with related studies on DOS estimation, multi-eigenvalue phase estimation, and modified QPE methods.

\subsection{Relation to Density of States Estimation}

The literature on density of states (DOS) estimation~\cite{goh2026dos,bai2026random} also uses randomized initial states, but its objective is different from ours.
Goh--Koczor~\cite{goh2026dos} estimates DOS from time evolution of simple random initial states, such as randomly chosen computational basis states, recovering DOS on average up to a Gaussian convolution window.
Bai et al.~\cite{bai2026random} develops random-state quantum algorithms for electronic-structure properties, combining Hadamard-test-based DOS estimation with a modified QPE method for local density of states (LDOS).

These works target coarse-grained spectral density or spatially resolved LDOS rather than simultaneous identification of all individual eigenphases.
By contrast, the present work focuses on the standard QPE output distribution itself and gives a peak-detection theory for locating individual eigenphases under the Fej\'er-kernel structure.

\subsection{Multi-Eigenvalue Phase Estimation}

There is also a broader literature on estimating multiple eigenvalues or phases from quantum measurement data.
O'Brien et al.~\cite{obrien2019quantum} studies low-cost QPE with a single ancilla qubit and compares time-series and Bayesian post-processing when the input state is not an eigenstate.
Dutkiewicz et al.~\cite{dutkiewicz2022heisenberg} develops Heisenberg-limited phase estimation of multiple eigenvalues with few control qubits, combining few-control-qubit phase-estimation routines with time-series or matrix-pencil subroutines.
Somma~\cite{somma2019eigenvalue} considers a different line of work: eigenvalue estimation from expectation values of the evolution operator via time-series analysis, explicitly contrasting this approach with QFT-based QPE.
Lim et al.~\cite{lim2024curvefitted} keeps the standard QPE circuit and improves precision by curve fitting to the full QPE output distribution, while also mentioning a possible extension to multiple phases.

These works either rely on post-processing of measurement data obtained from a fixed initial state or introduce alternative phase-estimation schemes.
By contrast, the present work addresses this initial-state-induced nonuniformity itself through per-shot randomization of the initial state and then performs eigenphase detection directly from the resulting state-averaged standard-QPE distribution via a simple QPE-only peak-detection rule with a shot-count guarantee.

\subsection{Modified QPE for LDOS}

Bai et al.~\cite{bai2026random} also introduces a modified QPE (M-QPE) method for real-space LDOS.
Its circuit prepares ancilla phases, applies controlled time evolutions, performs a QFT on the ancilla register, and then post-selects the ancilla outcome to project onto a chosen energy component.

The objective of M-QPE is to recover spatial distributions at selected energies, not to detect all eigenphases simultaneously.
By contrast, the present work does not modify the standard QPE circuit.
Instead, it uses per-shot randomization of the initial state to define a state-averaged QPE distribution and then detects all eigenphases through a simple peak-detection rule with an explicit shot-count guarantee.

\section{Preliminaries}\label{sec:preliminaries}

\subsection{Notation and Basic Setup}

Throughout this paper, let $m,n\in\mathbb{N}$ and set $M = 2^m$, $N = 2^n$.
Define $a_j = j/N$ for $j = 0, 1, \dots, N-1$, and write
$A_N = \{a_j \mid j = 0, 1, \dots, N-1\}$,
$\Theta = \{\theta_k \in [0,1) \mid k = 0, 1, \dots, M-1\}$,
$S_{\mathbb{C}}^{M-1} = \{z \in \mathbb{C}^{M} \mid \sum_{j=0}^{M-1}|z_j|^2 = 1\}$,
and denote by $U(M)$ the set of all $M \times M$ unitary matrices.
The torus $\mathbb{T} = \mathbb{R}/\mathbb{Z}$ is equipped with the distance $|x - y|_{\mathbb{T}} = \min_{k \in \mathbb{Z}} |x - y - k|$, and we regard $A_N$ and $\Theta$ as subsets of $[0,1) \subset \mathbb{T}$.

Let $X,Y,Z$ denote the $2 \times 2$ Pauli matrices, i.e.,
\begin{align*}
    X &=
    \begin{pmatrix}
        0 & 1 \\
        1 & 0
    \end{pmatrix}, \quad
    Y =
    \begin{pmatrix}
        0 & -i \\
        i & 0
    \end{pmatrix}, \quad
    Z =
    \begin{pmatrix}
        1 & 0 \\
        0 & -1
    \end{pmatrix}.
\end{align*}
$I$ denotes the identity matrix of appropriate size.

For $j \in \mathbb{Z}$, we extend periodically by setting $a_{j+N} = a_j$ and $\theta_{j+M} = \theta_j$.
For $a, b \in [0,1) \subset \mathbb{T}$, the open interval $(a, b)_{\mathbb{T}}$ is defined as the image under the natural projection $\mathbb{R} \to \mathbb{T}$:
\begin{align*}
    (a, b)_{\mathbb{T}}
    &=
    \begin{cases}
        (a, b), & a \leq b, \\[.5em]
        (a, 1+b), & a > b,
    \end{cases}
\end{align*}
and closed and half-open intervals are defined similarly.

For $k = 1, 2, \dots$, define $|0^k\rangle = |0\rangle^{\otimes k}$.
We write $\{I, X, Y, Z\}^{\otimes k}$ and $\{I, X\}^{\otimes k}$ for $\{\bigotimes_{j=1}^{k} P_j \mid P_j = I, X, Y, Z\}$ and $\{\bigotimes_{j=1}^{k} P_j \mid P_j = I, X\}$, respectively.

For a set $\mathcal{X}$, we denote its power set by $2^{\mathcal{X}}$.
When $\mathcal{X}$ is a topological space, we denote its Borel $\sigma$-algebra by $\mathcal{B}(\mathcal{X})$.
For notational convenience, we write $\mathcal{X}$ for the measurable spaces $(\mathcal{X}, \mathcal{B}(\mathcal{X}))$ and $(\mathcal{X}, 2^{\mathcal{X}})$ when there is no ambiguity.

Let $U \in U(M)$ have the eigendecomposition
\begin{align*}
    U = \sum_{k=0}^{M-1} e^{i2\pi\theta_k} |\psi_k\rangle\langle\psi_k|
\end{align*}
with eigenphases $\Theta$.

For $c = (c_0, \dots, c_{M-1}) \in S_{\mathbb{C}}^{M-1}$, define
\begin{align}
    |\psi\rangle
    &= \sum_{k=0}^{M-1} c_k |\psi_k\rangle. \label{eq:initial state}
\end{align}

\subsection{Quantum Phase Estimation (QPE)}

Let $|\psi\rangle$ be the initial state defined in~\eqref{eq:initial state}.
When standard QPE~\cite{nielsen2010quantum} with an $n$-bit ancilla register is applied to $|0^n\rangle \otimes |\psi\rangle$, the probability of obtaining measurement outcome $j = 0, 1, \dots, N-1$ on the ancilla register is
\begin{align}\label{eq:QPE distribution main}
    p_j = \frac{1}{N^2} \sum_{k=0}^{M-1} |c_k|^2 \, \frac{\sin^2 \pi N(\theta_k - a_j)}{\sin^2 \pi (\theta_k - a_j)}
\end{align}
(Proposition~\ref{proposition:QPE distribution}).
$N^{-2}\sin^2 \pi N(\theta_k - a_j)/\sin^2 \pi (\theta_k - a_j)$ above is the normalized Fej\'{e}r kernel $\tilde{F}_N(2\pi(\theta_k - a_j))$, which produces a peak of width $2/N$ centered at each eigenphase $\theta_k$.

\subsection{The Fej\'{e}r Kernel}

Define $F_N \colon \mathbb{R} \to \mathbb{R}$ by
\begin{align*}
    F_N(x) =
    \begin{cases}
        N, & x = 0, \pm 2\pi, \pm 4\pi, \dots \\[.5em]
        \displaystyle \frac{1}{N} \frac{\sin^2(Nx/2)}{\sin^2(x/2)}, & \text{otherwise}.
    \end{cases}
\end{align*}
Set $\tilde{F}_N(x) = F_N(x)/F_N(0) = F_N(x)/N$.
Then~\eqref{eq:QPE distribution main} can be written as 
\begin{align*}
    p_j &= \sum_{k=0}^{M-1} |c_k|^2 \tilde{F}_N(2\pi(\theta_k - a_j)).
\end{align*}

As a basic property of the Fej\'{e}r kernel, $\sum_{j=0}^{N-1} \tilde{F}_N(2\pi(\theta_k - a_j)) = 1$ holds (Lemma~\ref{lemma:Fejer properties}-(i)).
This is consistent with the fact that the QPE output forms a probability distribution.

\subsection{Block Encoding}\label{subsec:block encoding prelim}

For a positive semidefinite matrix $A$ with $\|A\| \leq 1$, a unitary matrix $U$ that has $A$ as its upper-left block is called a block encoding of $A$~\cite{GSLW19}. That is,
\begin{align}
    U = \begin{pmatrix} A & \cdot \\ \cdot & \cdot \end{pmatrix} \label{eq:def_U}
\end{align}
with $(\langle 0| \otimes I) U (|0\rangle \otimes I) = A$.
For concrete constructions of block encodings, see~\cite{martyn2021grand, GSLW19, low2019hamiltonian}.

\subsection{1-Design~\cite{mele2024haar}}

The set $\{I, X\}^{\otimes m}$ with equal probabilities $1/M$ constitutes a state 1-design (also called a spherical 1-design).
For all $k = 0, 1, \dots, M-1$,
\begin{align*}
    \frac{1}{M} \sum_{P \in \{I,X\}^{\otimes m}} |\langle \psi_k | P | 0^m \rangle|^2 = \frac{1}{M}
\end{align*}
holds (Lemma~\ref{lemma:Pauli sum}-(ii)).

The set $\{I, X, Y, Z\}^{\otimes m}$ with equal probabilities $1/M^2$ constitutes a unitary 1-design (Lemma~\ref{lemma:Pauli sum}(i)).

\section{State-Averaged QPE}\label{sec:state averaged QPE}

\subsection{QPE Measure and the State-Averaged Condition}

The output distribution~\eqref{eq:QPE distribution main} of QPE depends not only on the eigenphases $\Theta$ but also on the overlaps $c \in S_{\mathbb{C}}^{M-1}$ of the initial state.
When the goal is simultaneous detection of all eigenvalues, this dependence on the initial state becomes the main difficulty.
In this section, we introduce a probabilistic approach to averaging this dependence.

To describe the situation where the initial state is randomly selected at each shot, we introduce the concept of a QPE measure (Definition~\ref{definition:QPE measure}, Appendix).
A QPE measure $\mu$ is a probability measure on $A_N \times S_{\mathbb{C}}^{M-1}$ representing the joint distribution of the QPE measurement outcome $a_j$ and the coefficient vector $z=(z_0,\dots,z_{M-1})$ of the initial state.

\begin{definition}[State-averaged]\label{def:state averaged main}
    A $\nu$-QPE measure $\mu$ is said to be state-averaged if for all $k = 0, 1, \dots, M-1$,
    \begin{align*}
        \int |z_k|^2 \, \nu(dz) = \frac{1}{M}.
    \end{align*}
    Each $z_k$ corresponds to the overlap of the initial state with the $k$th eigenstate.
\end{definition}

Under the state-averaged condition, the measurement probability of each $a_j \in A_N$ under QPE becomes
\begin{align}\label{eq:state averaged distribution}
    p_j = \frac{1}{M} \sum_{k=0}^{M-1} \tilde{F}_N(2\pi(\theta_k - a_j)).
\end{align}
The right-hand side is independent of the initial state and is determined solely by the eigenphases $\Theta$.
All eigenvalues contribute to the Fej\'{e}r kernel peaks with equal weight $1/M$, enabling detection by a uniform threshold.

\subsection{Realization of the State-Averaged Condition via 1-Designs}

The state-averaged condition can be easily realized using a 1-design.

For each element $P \in \{I, X\}^{\otimes m}$, select the initial state $P|0^m\rangle$ with equal probability $1/M$.
As stated above, for all $k = 0, 1, \dots, M-1$,
\begin{align*}
    \frac{1}{M} \sum_{P \in \{I,X\}^{\otimes m}} |\langle \psi_k | P | 0^m \rangle|^2 = \frac{1}{M}
\end{align*}
holds. This is precisely the state-averaged condition.

In practice, the implementation requires only choosing, for each shot, whether to apply a bit flip ($X$ operator) to each of the $m$ qubits independently with equal probability.
This is equivalent to randomly selecting one computational basis state $|j\rangle$ ($j = 0, 1, \dots, M-1$), and the additional quantum circuit cost is negligible.

\begin{remark}
    It is also possible to construct a state-averaged QPE measure using the Haar measure; however, initial state preparation based on the Haar measure is known to be inefficient~\cite{mele2024haar}.
    A 1-design reproduces only the first moment of the Haar measure and is the simplest construction; this suffices for the state-averaged condition.
\end{remark}

\section{Peak Detection and Shot Count Estimation}\label{sec:peak detection}

\subsection{Theory of Peak Detection}

We present the theoretical foundation for detecting all distinct eigenvalues under the state-averaged QPE measure.

For each eigenphase $\theta_k$, define its $1/N$-neighborhood 
\begin{align*}
    I_k = \left[ \theta_k - \frac{1}{N},\, \theta_k + \frac{1}{N} \right]_{\mathbb{T}},
\end{align*}
and define the detection set 
\begin{align*}
    C = \left\{ a_j \in A_N \,\middle|\, p_j \geq \frac{\tau}{M},~j = 0, 1, \dots, N-1 \right\}, \quad \tau = \frac{4}{\pi^2}.
\end{align*}

Theorem~\ref{theorem:peak detection} (Appendix) guarantees the following when the minimum gap between adjacent eigenphases is at least $3/N$:
\begin{itemize}
    \item[(1)] $C$ intersects every eigenphase neighborhood $I_k$ (no missed detection).
    \item[(2)] $C$ is contained in $\bigcup_k I_k$ (no false positives).
    \item[(3)] The number of points of $C$ belonging to the same $I_k$ is at most two.
\end{itemize}

These properties imply that eigenphases can be uniquely estimated from the grid points contained in $C$.

\subsection{Threshold Gap}

A positive gap between the measurement probabilities of $a_j \in C$ and $a_j \notin C$ is highly useful for the stability of peak detection.

Indeed, the analysis in the Appendix (Theorem~\ref{theorem:peak detection}\eqref{eq:gap between C and not C}) shows that for $a_j \notin \bigcup_k I_k$, we have $p_j \leq \sigma/M+d_N$, where
\begin{align*}
    \sigma = \frac{2}{\pi^2} \sum_{l=0}^{\infty} \frac{1}{(3l+1)^2}, \quad d_N=\frac{1-\tau}{N^2}
\end{align*}
and $\sigma/M+d_N < \tau/M$.
That is, there exists a gap $(\tau - \sigma)/M-d_N> 0$ between the measurement frequencies of $a_j$ near peaks and $a_j$ far from peaks.

\subsection{Shot Count Estimate}

Theorem~\ref{theorem:shot estimate} (Appendix) uses this gap to provide an estimate of the shot count sufficient for the threshold test on the empirical distribution $\hat{p}_j$ to function correctly for all $a_j$.

Let $\varepsilon>0$ satisfy
    \begin{align*}
        0 < \varepsilon \leq \frac{\tau - \sigma}{2} - \frac{Md_N}{2},
    \end{align*}
and $0 < \delta < 1$.
When the shot count $K$ satisfies
\begin{align*}
    K
    &\geq \left\{H^{+}\left(\frac{\gamma}{M}+d_N, \frac{\varepsilon}{M}\right)\right\}^{-1}\log\frac{N+M}{\delta} \\[.5em]
    & \sim \left\{(\gamma+\varepsilon)\log\frac{\gamma+\varepsilon}{\gamma} - \varepsilon \right\}^{-1}\,M\log\frac{N+M}{\delta} \quad (M \to \infty),
\end{align*}
all peaks are correctly detected with probability at least $1 - \delta$.
Here,
\begin{align*}
    H^+(x, a) &= (x+a)\log\frac{x+a}{x} + (1-x-a)\log\frac{1-x-a}{1-x}
\end{align*}
is the Kullback--Leibler divergence between Bernoulli distributions with parameters $x+a$ and $x$, and
\begin{align*}
    \gamma &= 1 + \frac{1}{\pi^2}\left( \frac{\pi^2}{6} - \sum_{l=0}^\infty \frac{1}{(3l+1)^2} \right).
\end{align*}

This estimate shows that the required number of shots is $O(M\log(N+M))$.

\section{Numerical Experiments}\label{sec:numerical experiments}

\subsection{Problem Setup}\label{subsec:problem setup}

We consider an eigenvalue problem arising from CAE.
Specifically, we address a structural natural frequency analysis and consider the eigenvalue problem discretized by the finite element method (FEM).

\subsubsection{Analytical Model}

The target structure is a cantilever beam.
The beam dimensions are $1000\,\mathrm{mm} \times 200\,\mathrm{mm} \times 100\,\mathrm{mm}$ (length $\times$ width $\times$ height), with material constants: Young's modulus $E = 205{,}000\,\mathrm{MPa}$, Poisson's ratio $\nu = 0.3$, and density $\rho = 7.85 \times 10^{-9}\,\mathrm{kg/mm^3}$.
These correspond to typical structural steel.

The discretization uses first-order hexahedral elements with a $16 \times 6 \times 2$ mesh in the length, width, and height directions, respectively.
As a boundary condition, all degrees of freedom on the face at $x = 0$ are constrained (fixed end).
The total number of degrees of freedom is $1{,}071$, and the number after applying boundary conditions is $m_0 = 1{,}008$.

A lumped mass matrix is adopted for the mass matrix.
The lumped mass matrix $\boldsymbol{\mathrm{M}}$ is a positive definite diagonal matrix, and the stiffness matrix $\boldsymbol{\mathrm{K}}$ is a positive definite real symmetric matrix.

\begin{figure}[ht]
    \centering
    \includegraphics[width=\linewidth]{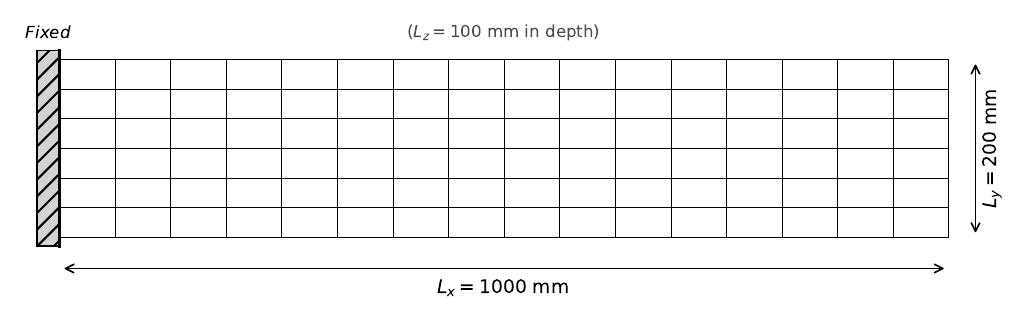}
    \caption{Finite element mesh of the cantilever beam ($16 \times 6 \times 2$ hexahedral elements, $m_0 = 1{,}008$ DOFs after boundary conditions).}
    \label{fig:fem_model}
\end{figure}

\subsubsection{Formulation of the Eigenvalue Problem}

Using $\boldsymbol{\mathrm{M}}$ and $\boldsymbol{\mathrm{K}}$, consider the generalized eigenvalue problem
\begin{align}
    \boldsymbol{\mathrm{K}}\boldsymbol{\mathrm{v}} = \lambda \boldsymbol{\mathrm{M}}\boldsymbol{\mathrm{v}}, \label{eq:GEVP1}
\end{align}
where $\lambda$ is an eigenvalue and $\boldsymbol{\mathrm{v}}$ is the corresponding eigenvector.
By substituting $\boldsymbol{\mathrm{u}} = \boldsymbol{\mathrm{M}}^{1/2}\boldsymbol{\mathrm{v}}$ in~\eqref{eq:GEVP1}, the problem reduces to the standard eigenvalue problem
\begin{align*}
    \boldsymbol{\mathrm{A}}\boldsymbol{\mathrm{u}} = \lambda \boldsymbol{\mathrm{u}},
\end{align*}
where $\boldsymbol{\mathrm{A}} = \boldsymbol{\mathrm{M}}^{-1/2}\boldsymbol{\mathrm{K}}\boldsymbol{\mathrm{M}}^{-1/2}$ is a positive definite matrix.

Let $\boldsymbol{\mathrm{A}} = (a_{kj})_{k,j=0}^{m_0-1}$, and extend it by appending zero entries to form an $M \times M$ ($M = 2^m$) matrix $A$:
\begin{align*}
    A
    &=
    \left(
        \begin{array}{c:c}
        {
            \begin{array}{ccc}
            a_{00} & \cdots & a_{0,m_0-1}  \\
                \vdots & \ddots & \vdots \\
                a_{m_0-1,0} & \cdots & a_{m_0-1,m_0-1}
            \end{array}
        } & 0 \\
            \hdashline
        0 & 
            \begin{array}{ccc}
                0 &  & \\
                & \ddots & \\
                & & 0
            \end{array}
        \end{array}
    \right).
\end{align*}
The matrix $A$ is positive semidefinite and real symmetric, with $\mathrm{ker}\,A = \mathrm{span}\{|m_0\rangle, \dots, |M-1\rangle\}$.

The eigendecomposition of $A$ is
\begin{align*}
    A = \sum_{k=0}^{m_0-1} \lambda_k |\psi_k\rangle\langle\psi_k| + \sum_{k=m_0}^{M-1} 0 \cdot |k\rangle\langle k|.
\end{align*}
Choosing $\alpha > 0$ such that $\|A/\alpha\| < 1$ and replacing $A$ by $A/\alpha$, we may assume $0 < \lambda_k < 1$ for $k = 0, 1, \dots, m_0-1$.
For $k = m_0, \dots, M-1$, set $\lambda_k = 0$ and $|\psi_k\rangle = |k\rangle$.

\subsection{Block Encoding}

For the matrix $A$ constructed in Section~\ref{subsec:problem setup}, there exists a unitary $U$ that provides a $(1,1,0)$-block encoding of $A$~\cite{GSLW19}.
In this numerical experiment, following~\cite{martyn2021grand}, we define $U$ using the eigendecomposition of $A$ as:
\begin{align*}
    U
    &= \sum_{k=0}^{M-1} 
    \begin{pmatrix}
        \lambda_k & \sqrt{1-\lambda_k^2} \\[.5em]
        \sqrt{1-\lambda_k^2} & -\lambda_k
    \end{pmatrix}
    \otimes |\psi_k\rangle\langle\psi_k| \\[.5em]
    &=
    \begin{pmatrix}
        A & \sqrt{I-A^2} \\[.5em]
        \sqrt{I-A^2} & -A
    \end{pmatrix}.
\end{align*}

In general, the off-diagonal blocks of a unitary satisfying $(\langle 0| \otimes I)\,U\,(|0\rangle \otimes I) = A$ are not unique.
However, the measurement distribution of QPE depends only on the weights $|c_k|^2$ of the initial state with respect to each eigenspace and not on the specific construction of the block encoding.
Therefore, the results of the following numerical experiments do not depend on the particular form~\eqref{eq:def_U}.

We note that the numerical experiments in this paper are based on the mathematical analysis of the QPE output distribution and do not address the quantum circuit implementation of $U$.

\subsection{QPE Simulation}\label{subsec:qpe simulation}

For $k = 0, 1, \dots, M-1$, let $\theta_k \in [0, \pi/2]$ satisfy $\lambda_k = \cos\theta_k$.
Then
\begin{align*}
    U(Z \otimes I)
    &= \sum_{k=0}^{M-1} e^{-i\theta_k Y} \otimes |\psi_k\rangle\langle \psi_k| \\
    &= \sum_{k=0}^{m_0-1} e^{-i\theta_k}|y_+\rangle\langle y_+| \otimes |\psi_k\rangle\langle \psi_k| + \sum_{k=0}^{m_0-1} e^{i\theta_k}|y_-\rangle\langle y_-| \otimes |\psi_k\rangle\langle \psi_k| \\[.5em]
    & \hspace{2em} + \sum_{k=m_0}^{M-1} e^{-i\frac{\pi}{2}}|y_+\rangle\langle y_+| \otimes |k\rangle\langle k| + \sum_{k=m_0}^{M-1} e^{i\frac{\pi}{2}}|y_-\rangle\langle y_-| \otimes |k\rangle\langle k|,
\end{align*}
where
\begin{align*}
    |y_\pm\rangle = \frac{|0\rangle \pm i \,|1\rangle}{\sqrt{2}}.
\end{align*}

The eigenvalues of $\tilde{U} := U(Z \otimes I)$ are $e^{\pm i\theta_k}$, with corresponding eigenstates $|y_\mp\rangle \otimes |\psi_k\rangle$.
By choosing the ancilla initial state as $|y_-\rangle$, only the phases $+\theta_k$ (corresponding to eigenvalues $e^{i\theta_k}$) are measured, while the phases $-\theta_k$ are not.

The target phases in the numerical experiment are the positive eigenphases of the original operator $\boldsymbol{\mathrm{A}}$.
As explained above, these phases satisfy $\theta_k \in (0,\pi/2)$ and hence, after normalization, lie in $(0,1/4)$, so their measurement bit strings begin with $0.00$.
By contrast, the phases introduced by zero padding correspond to $\pm\pi/2$, namely to the bit strings $0.01$ and $0.11$, while the negative phases lie in $(3/4,1)$.
Therefore, the target phases are separated from both the padding-induced phases and the negative phases at the measurement level.
Under this separation, if the original operator $\boldsymbol{\mathrm{A}}$ has no repeated eigenvalues, then the distinct-phase assumption used in the theoretical analysis is valid for the target phases considered here.

The initial state is prepared by selecting from $\{|y_-\rangle \otimes |j\rangle \mid j = 0, 1, \dots, m_0-1\}$ with equal probability $1/m_0$.
QPE then estimates the phase $\theta_k/(2\pi)$.
Since $\theta_k \in (0, \pi/2)$ implies $\theta_k/(2\pi) \in (0, 1/4)$, starting the phase kickback from $\tilde{U}^{2^2}$ yields $\vartheta_k := 4 \cdot \theta_k/(2\pi) \in (0, 1)$ as the measured quantity.

Let $n$ be the number of ancilla qubits for QPE measurement and set $N = 2^n$.
By~\cite{lim2024curvefitted} and Proposition~\ref{proposition:QPE distribution}, the probability of measuring $a_j$ in a single QPE execution with initial state $|j_0\rangle$ is
\begin{align*}
    p_{j\mid j_0}
    &= \sum_{k=0}^{m_0-1} |\langle j_0| \psi_k \rangle|^2 \, \tilde{F}_N(2\pi(\vartheta_k - a_j))
\end{align*}
for $j = 0, 1, \dots, N-1$.

\subsection{Detection Procedure}\label{subsec:detection procedure}

In the numerical experiment, the initial state is sampled uniformly from
$\{|y_-\rangle \otimes |j\rangle \mid j=0,1,\dots,m_0-1\}$.
Hence, the empirical distribution is normalized by $m_0$, and the detection threshold is scaled by $1/m_0$ rather than $1/M$.

For reporting the detection performance, let $r$ denote the number of target eigenphases considered in the experiment.
In the present cantilever example, we take $r=m_0=1{,}008$.

The set $C$ in Theorem~\ref{theorem:peak detection} is defined by the true probabilities $p_j$ and therefore cannot be directly constructed in experiments.
Instead, setting 
\begin{align*}
    \varepsilon &= \frac{\tau-\sigma}{2} - \frac{m_0d_N}{2},
\end{align*}
we construct the estimated peak set from the empirical distribution $\hat{p}_j$ using the threshold 
\begin{align*}
    \frac{\tau}{m_0} - \frac{\varepsilon}{m_0} = \frac{\tau+\sigma}{2m_0} + \frac{d_N}{2}.
\end{align*}
Thus, we define
\begin{align*}
    \hat{C} = \left\{a_j \in A_N \,\middle|\, \hat{p}_j \geq \frac{\tau+\sigma}{2m_0} + \frac{d_N}{2},~j = 0, 1, \dots, N-1\right\}.
\end{align*}

Applying Theorem~\ref{theorem:shot estimate} to the target phases considered here, with $M$ replaced by $r$ (here $r=m_0$), that is, to the output distribution with equal weights $1/m_0$ on the subspace spanned by the target eigenstates $\mathrm{span}\{|y_-\rangle \otimes |j\rangle \mid j=0,1,\dots,m_0-1\}$ (see also Remark \ref{remark:zero-padding case}), we obtain that, with sufficiently many shots, $\hat{C}$ satisfies with high probability
\begin{align*}
    C \subset \hat{C} \subset \bigcup_{k=0}^{r-1} I_k,
\end{align*}
and properties (1)--(4) of Theorem~\ref{theorem:peak detection from samples} hold.
That is, $\hat{C}$ contains all peaks near eigenphases without omission and does not contain any points far from eigenphases.
Note that, by Theorem~\ref{theorem:peak detection from samples}, any consecutive run of grid points contained in $\hat{C}$ has length at most three.
Accordingly, the procedure for obtaining eigenphase estimates from $\hat{C}$ is as follows.
\begin{description}
    \item[Estimation~1] If $a_{j-1},a_j,a_{j+1}\in\hat{C}$, set $\hat{\vartheta}=a_j$.
    If $a_j \in \hat{C}$, $a_{j+1} \in \hat{C}$, and $a_{j-1},a_{j+2}\notin\hat{C}$, the weighted interpolation
    \begin{align*}
        \hat{\vartheta} = \frac{\hat{p}_j \cdot a_j + \hat{p}_{j+1} \cdot a_{j+1}}{\hat{p}_j + \hat{p}_{j+1}}
    \end{align*}
    is taken as the eigenphase estimate.
    \item[Estimation~2] If $a_j \in \hat{C}$ and $a_{j-1}, a_{j+1} \notin \hat{C}$, set $\hat{\vartheta} = a_j$ as the estimate.
\end{description}

From the estimated phase $\hat{\vartheta}$, the natural frequency is recovered as
\begin{align*}
    \hat{f} = \frac{\sqrt{\alpha\cos(2\pi\hat{\vartheta}/4)}}{2\pi},
\end{align*}
where $\alpha>0$ is the scaling factor introduced above.

Let $\hat{r}$ denote the number of eigenphase estimates obtained from $\hat{C}$ via Estimation~1 and Estimation~2.
We say that $\hat{r}$ target eigenphases are detected.
When $\hat{r} < r$, we say that $r-\hat{r}$ target eigenphases are undetected.
When $\hat{r} = r$, we say that all target eigenphases are detected.
The ratio $\hat{r}/r$ is called the detection rate.

\subsection{Experimental Design}\label{subsec:experimental design}

The numerical experiments consist of the following two experiments to validate the theoretical claims.

\subsubsection{Experiment~1: Simultaneous Detection of All Eigenvalues}

We verify that all eigenvalues of the cantilever beam model are detected with the shot count based on the estimate from Theorem~\ref{theorem:shot estimate}.
The parameters are as follows.
\begin{itemize}
    \item Degrees of freedom: $m_0 = 1{,}008$
    \item Number of ancilla qubits: $n = 27$ ($N = 134{,}217{,}728$)
    \item Confidence parameter: $\delta = 0.001$
    \item Shot count: $K = 7{,}060{,}000$ ($\geq 7{,}052{,}323$ by Theorem~\ref{theorem:shot estimate})
\end{itemize}
The value $n = 27$ was chosen as the minimum $n$ satisfying $3/N < \min_k |\vartheta_{k+1} - \vartheta_k|_{\mathbb{T}}$ for the minimum adjacent eigenphase gap $\min_k |\vartheta_{k+1} - \vartheta_k|_{\mathbb{T}} = 3.58 \times 10^{-8}$.

The experiment is run at five shot count levels: $0.25$, $0.5$, $0.75$, $1.0$, and $1.5$ times the theoretical estimate.

\subsubsection{Experiment~2: Search for the Critical Shot Count}

To quantitatively evaluate the conservativeness of the theoretical shot count estimate, we vary the shot count from $0.005$ to $0.1$ times the baseline shot count $K$ and identify the critical point at which the detection rate drops below 100\%.
For each shot count, three trials with different random seeds are performed to confirm the stability of the detection rate.

\subsection{Results}\label{subsec:results}

\subsubsection{Results of Experiment~1}

Table~\ref{tab:experiment1} shows the detection results at each shot count level.
All $1{,}008$ target eigenphases were detected at $100\%$ across all levels from $0.25K$ to $1.5K$.

\begin{table}[ht]
\centering
\caption{Experiment~1: Shot count and detection results ($m_0 = 1{,}008$, $n = 27$)}
\label{tab:experiment1}
\begin{tabular}{crcccc}
\hline
Fraction & Shots & Detected & Detection Rate & Phase RMSE & Freq.\ Rel(Max) \\
\hline
$0.25K$ & $1{,}765{,}000$  & 1{,}008/1{,}008 & 100\% & $1.78 \times 10^{-9}$ & $9.26 \times 10^{-5}$ \\
$0.50K$ & $3{,}530{,}000$  & 1{,}008/1{,}008 & 100\% & $1.76 \times 10^{-9}$ & $8.02 \times 10^{-5}$ \\
$0.75K$ & $5{,}295{,}000$  & 1{,}008/1{,}008 & 100\% & $1.79 \times 10^{-9}$ & $7.05 \times 10^{-5}$ \\
$1.0K$  & $7{,}060{,}000$  & 1{,}008/1{,}008 & 100\% & $1.78 \times 10^{-9}$ & $1.07 \times 10^{-4}$ \\
$1.5K$  & $10{,}590{,}000$ & 1{,}008/1{,}008 & 100\% & $1.77 \times 10^{-9}$ & $6.69 \times 10^{-5}$ \\
\hline
\end{tabular}
\end{table}

The phase estimation RMSE remains in the narrow range from $1.76 \times 10^{-9}$ to $1.79 \times 10^{-9}$ across all shot count levels.
Likewise, the maximum relative error of the recovered natural frequency remains stable, staying between $6.69 \times 10^{-5}$ and $1.07 \times 10^{-4}$.

These results indicate that, for the present cantilever example, the simultaneous detection of all target eigenphases is already stable at $0.25K$.
Over the tested range, increasing the shot count yields no improvement in the detection rate and only minor variations in the final error metrics.

\subsubsection{Results of Experiment~2}

Table~\ref{tab:experiment2} shows the results of the critical shot count search.

\begin{table}[ht]
\centering
\caption{Experiment~2: Search for the critical shot count ($m_0 = 1{,}008$, $n = 27$, 3 trials each)}
\label{tab:experiment2}
\begin{tabular}{crcc}
\hline
Fraction & Shots & Detection Rate (Mean) & Detection Rate (Min) \\
\hline
$0.005K$ & $35{,}300$      & 99.07\% & 98.71\% \\
$0.01K$  & $70{,}600$      & 99.83\% & 99.70\% \\
$0.02K$  & $141{,}200$     & 100.0\% & 100.0\% \\
$0.03K$  & $211{,}800$     & 100.0\% & 100.0\% \\
$0.05K$  & $353{,}000$     & 100.0\% & 100.0\% \\
$0.10K$  & $706{,}000$     & 100.0\% & 100.0\% \\
\hline
\end{tabular}
\end{table}

The critical shot count at which all target eigenphases are stably detected is
\begin{align*}
    K_{\mathrm{critical}} \approx 141{,}200 = 0.02K,
\end{align*}
since from this level onward all three trials achieve a detection rate of $100\%$.

Thus, the practical critical shot count is approximately $1/50$ of the baseline shot count $K=7{,}060{,}000$, and is also very close to $1/50$ of the theoretical estimate from Theorem~\ref{theorem:shot estimate}.
This shows that the theoretical shot count estimate functions correctly as a sufficient condition, but is considerably conservative for the present cantilever example.

The transition is gradual rather than abrupt.
Even at $0.005K = 35{,}300$ shots, the detection rate remains between $98.7\%$ and $99.4\%$, corresponding to $995$ to $1{,}002$ detected target eigenphases out of $1{,}008$.
At $0.01K = 70{,}600$ shots, one of the three trials already achieves full detection, and the minimum detection rate is still as high as $99.7\%$.
Hence, the loss of complete detection in the low-shot regime is not a global failure of the method, but is caused by only a very small number of modes near the detection threshold.

\subsection{Discussion of Experiments}\label{subsec:discussion experiments}

The results of Experiments~1 and~2 confirm the following.
\begin{itemize}
    \item[(1)] Using standard QPE with 1-design randomized initial states, all $m_0 = 1{,}008$ eigenvalues were simultaneously detected.
    This demonstrates that all eigenvalues are accessible by random selection of computational basis states alone, without individually preparing each eigenstate.

    \item[(2)] Once detection succeeds, the phase estimation accuracy becomes nearly independent of the shot count.
    The residual phase error is governed primarily by the fixed ancilla resolution together with the interpolation procedure, rather than by finite-shot statistical fluctuations.

    \item[(3)] The theoretical shot count estimate (Theorem~\ref{theorem:shot estimate}) functions correctly as a sufficient condition.
    On the other hand, the experimentally observed critical shot count is approximately $1/50$ of the baseline shot count and also very close to $1/50$ of the theoretical estimate, indicating that the theory is conservative.
    Moreover, the degradation below the full-detection threshold is gradual rather than abrupt: even at $0.005K$, more than $98.7\%$ of the target eigenphases are still detected.
    This suggests that low-shot failures are caused only by a very small number of modes near the detection threshold, rather than by a global breakdown of the method.
    One possible explanation for this conservativeness is that, in the proof of Theorem~\ref{theorem:shot estimate}, the effective threshold gap is derived from the universal bounds in Theorem~\ref{theorem:peak detection}, namely the lower bound $\tau/M$ on peak bins and the upper bound $\sigma/M+d_N$ on non-peak bins, and this derived gap may be substantially smaller than the actual probability gap in the present example.
\end{itemize}

\section{Discussion and Conclusion}\label{sec:conclusion}

\subsection{Summary of Main Results}

In this paper, we have shown that standard QPE with 1-design randomized initial states makes all eigenvalues accessible at the level of the state-averaged QPE distribution.
This is achieved without any modification to the standard QPE circuit, simply by selecting the initial state randomly from the computational basis at each shot and analyzing the resulting output distribution and its sampled peak structure.

Specifically, we have constructed the following theoretical framework:
\begin{itemize}
    \item Definition of the state-averaged QPE measure and its realization via 1-designs (Section~\ref{sec:state averaged QPE}).
    \item A peak detection theory for distinct eigenphases based on properties of the Fej\'er kernel (Theorem~\ref{theorem:peak detection}).
    \item A sample-based peak detection result showing that the detected set $\hat{C}$ obtained from finitely many shots still identifies the true peaks in a rigorous sense (Theorem~\ref{theorem:peak detection from samples}).
    \item An estimate of the shot count sufficient for detection of all distinct eigenvalue peaks (Theorem~\ref{theorem:shot estimate}).
\end{itemize}

These theoretical results were validated through numerical experiments on an FEM matrix with $1{,}008$ degrees of freedom.
In the present example, all target eigenvalues were detected at 100\%, confirming that the proposed sample-based peak-detection procedure works effectively in practice and that the theoretical shot-count estimate functions as a sufficient condition.

\subsection{Trade-off between Shot Cost and Initial State Preparation}

In many existing applications of QPE, eigenvalues can be estimated from a small number of shots by preparing an initial state with sufficient overlap to the target eigenstate.
However, this initial state preparation is itself a generally difficult quantum task.

In our approach, the difficulty of initial state preparation is eliminated, at the cost of requiring $O(M\log(N+M))$ shots.
Under the state-averaged condition, the peak height near each eigenphase is approximately $1/M$, and this order of shots appears necessary to detect all peaks.
This can be viewed as a structure in which the quantum cost of initial state preparation is transferred to a classical cost in shot count.

Since the matrix dimension satisfies $M = 2^m$ for the number of system qubits $m$, the shot count $O(M\log(N+M))$ grows exponentially in $m$.
However, in typical CAE/FEM applications, $M$ is chosen as the smallest power of two exceeding the physical degrees of freedom $m_0$, and $M$ and $m_0$ are of the same order (in our experiment, $m_0 = 1{,}008$ and $M = 1{,}024$).
The scalability of this approach is therefore governed by the physical problem size $m_0$ rather than directly by the number of qubits, though the exponential dependence on $m$ imposes a practical upper limit on the tractable problem size.

Furthermore, even when only eigenvalues in a specific frequency band are needed, shots corresponding to all eigenphases are required, and shots outside the target band are discarded in classical post-processing.
Improving this shot efficiency is an important open problem.

\subsection{Possible Applications}

Beyond its use for eigenvalue detection itself, state-averaged QPE may also be useful for obtaining prior information on eigenvalue locations for other quantum algorithms.
One possible application is as a preprocessing step for HHL.
If state-averaged QPE is performed beforehand and the phase-register values corresponding to genuine eigenphase peaks are identified, then one can restrict the controlled rotations in HHL to those detected phase-register values, rather than assigning them to all phase-register values.
This reduces the number of controlled-rotation cases to be treated explicitly and may also suppress unnecessary rotations associated with values not corresponding to genuine eigenphase peaks.

\subsection{Future Directions}\label{sec:future directions}

We identify the following directions for future work.

First, improvement of the shot count estimate.
The experimentally observed critical shot count is approximately $1/50$ of the theoretical estimate, indicating that the sufficient condition of Theorem~\ref{theorem:shot estimate} is conservative.
A tighter estimate is desirable.

Second, improvement of shot efficiency for band-selective detection.
In our approach, all shots are used to sample the full state-averaged QPE distribution even when only eigenvalues in a specific band are needed.
A combination with band filtering and amplitude amplification via QSVT is conceivable; however, whether this is compatible with the state-averaging procedure is not clear at present and requires further investigation.

Third, extension to repeated eigenphases.
Although repeated eigenphases remain accessible in the state-averaged QPE distribution through the aggregated weight of their eigenspaces, the present rigorous peak-detection theory is stated for distinct eigenphases, and extending it to the repeated-eigenphase case is nontrivial.
Under the state-averaged condition, multiplicities amplify both the peak and off-peak contributions in the QPE distribution.
At the same time, the inverse-square sidelobe decay of the Fej\'er kernel suggests that a local multiplicity-aware extension may still be possible under an appropriate separation condition.
We leave this problem for future work.

\startappendix
\section*{Appendix: Proofs and Technical Details}
In this appendix, we provide detailed proofs and technical supplements for the claims stated in the main text.
Since the peak-detection method developed in this paper is based on properties of the Fej\'er kernel, we also present several facts used in the peak-detection procedure that were not explained in detail in the main text.

Let $U \in U(M)$ have the eigendecomposition
\begin{align*}
    U = \sum_{k=0}^{M-1} e^{i2\pi\theta_k} |\psi_k\rangle\langle\psi_k|
\end{align*}
with eigenphases $\Theta$.
For $c = (c_0, \dots, c_{M-1}) \in S_{\mathbb{C}}^{M-1}$, define
\begin{align*}
    |\psi\rangle
    &= \sum_{k=0}^{M-1} c_k |\psi_k\rangle.
\end{align*}
The initial state including the ancilla register for QPE is $|0^n\rangle \otimes |\psi\rangle$.
The following is known regarding the state and measurement probabilities after QPE~\cite{lim2024curvefitted}.

\begin{proposition}[QPE Distribution]\label{proposition:QPE distribution}
    Let $U_{\textup{QPE}}$ denote the unitary operation implementing QPE on $|0^n\rangle \otimes |\psi\rangle$, and set $|\phi\rangle = U_{\textup{QPE}}\,(|0^n\rangle \otimes |\psi\rangle)$.
    Then
    \begin{align*}
        |\phi\rangle
        &= \frac{1}{N} \sum_{k=0}^{M-1}\sum_{j=0}^{N-1} c_k\, \frac{\sin \pi N \left( \theta_k - a_j\right)}{\sin \pi \left( \theta_k - a_j \right)}\, e^{i\pi (N-1) \left( \theta_k - a_j \right)}\, |j\rangle \otimes |\psi_k\rangle.
    \end{align*}
    The probability of measuring outcome $j = 0, 1, \dots, N-1$ is
    \begin{align*}
        \langle \phi|H_j|\phi\rangle
        &= \frac{1}{N^2} \sum_{k=0}^{M-1} |c_k|^2 \, \frac{\sin^2 \pi N \left( \theta_k - a_j \right)}{\sin^2 \pi \left( \theta_k - a_j \right)},
    \end{align*}
    where $H_j = |j\rangle\langle j| \otimes I$.
\end{proposition}

This leads to the following definition.

\begin{definition}\label{definition:QPE measure}
    A probability measure $\mu$ on $A_N \times S_{\mathbb{C}}^{M-1}$ is called a QPE measure with distribution $\nu$ on $S_{\mathbb{C}}^{M-1}$ (or simply a $\nu$-QPE measure, or QPE measure) if
    \begin{align*}
        \mu(\{a\} \times B)
        &= \frac{1}{N^2}\int_B \nu(dz)\sum_{k=0}^{M-1}|z_k|^2 \, \frac{\sin^2 \pi N \left( \theta_k - a \right)}{\sin^2 \pi \left( \theta_k - a \right)}, \quad a \in A_N,~B \in \mathcal{B}(S_{\mathbb{C}}^{M-1}),
    \end{align*}
    where $z = (z_0, \dots, z_{M-1}) \in S_{\mathbb{C}}^{M-1}$.

    Furthermore, $\mu$ is said to be state-averaged if for all $k = 0, 1, \dots, M-1$,
    \begin{align*}
        \int_{S_{\mathbb{C}}^{M-1}} |z_k|^2 \, \nu(dz) = \frac{1}{M}.
    \end{align*}
\end{definition}

We restrict attention hereafter to the case where the eigenphases of $U$ are distinct, i.e.,
\begin{align*}
    0 \leq \theta_0 < \theta_1 < \cdots < \theta_{M-1} < 1.
\end{align*}
The case of repeated eigenphases is beyond the scope of the present paper and is discussed as a future direction in Section~\ref{sec:future directions}.

\subsection*{Examples of State-Averaged QPE Measures}
We present practical examples of state-averaged QPE measures for initial state preparation.
\begin{lemma}\label{lemma:Pauli sum}
    Let $\mathcal{P}_{m} = \{I, X, Y, Z\}^{\otimes m}$ and $\mathcal{X}_{m} = \{I, X\}^{\otimes m}$.
    Then the following (i) and (ii) hold.
    \begin{itemize}
        \item[(i)] For any $|x\rangle, |y\rangle \in \mathbb{C}^M$ with $\langle x|x\rangle = \langle y|y\rangle = 1$,
            \begin{align*}
                \frac{1}{M^2}\sum_{P \in \mathcal{P}_{m}} |\langle x|P|y\rangle|^2
                &= \frac{1}{M}.
            \end{align*}
        \item[(ii)] For any $|x\rangle \in \mathbb{C}^M$ with $\langle x|x\rangle = 1$,
            \begin{align*}
                \frac{1}{M}\sum_{P \in \mathcal{X}_{m}} |\langle x|P|0^{m}\rangle|^2
                &= \frac{1}{M}.
            \end{align*}
    \end{itemize}
\end{lemma}
\begin{proof}
    (i) For $A \in \mathbb{C}^{M\times M}$, by the Pauli twirl identity~\cite[4.7.4]{wilde2017quantum},
    \begin{align*}
        \frac{1}{M^2}\sum_{P \in \mathcal{P}_{m}} P A P
        &= \frac{\mathrm{tr}\,A}{M}I.
    \end{align*}

    Setting $A = |y\rangle\langle y|$,
    \begin{align*}
        \frac{1}{M^2}\sum_{P \in \mathcal{P}_{m}} |\langle x|P|y\rangle|^2
        &= \frac{1}{M^2}\sum_{P \in \mathcal{P}_{m}} \langle x|P A P|x\rangle \\
        &= \frac{1}{M}.
    \end{align*}

    (ii) This follows from $\sum_{P \in \mathcal{X}_{m}} |\langle x|P|0^{m}\rangle|^2 = \sum_{j=0}^{M-1} |\langle x|j\rangle|^2 = 1$.
\end{proof}

Define measures $\tilde{\nu}_1$ and $\tilde{\nu}_2$ on the measurable spaces $(\mathcal{P}_{m}, 2^{\mathcal{P}_{m}})$ and $(\mathcal{X}_{m}, 2^{\mathcal{X}_{m}})$ by
\begin{align*}
    & \tilde{\nu}_1(\{P\}) = \frac{1}{M^2}, \quad P \in \mathcal{P}_{m}, \\[.5em]
    & \tilde{\nu}_2(\{P\}) = \frac{1}{M}, \quad P \in \mathcal{X}_{m},
\end{align*}
respectively.
Define the maps
\begin{align*}
    & \phi_1 \colon \mathcal{P}_{m} \ni P \mapsto (z_0, \dots, z_{M-1}) = (\langle \psi_0|P|0^{m}\rangle, \dots, \langle \psi_{M-1}|P|0^{m}\rangle) \in S_{\mathbb{C}}^{M-1}, \\
    & \phi_2 \colon \mathcal{X}_{m} \ni P \mapsto (z_0, \dots, z_{M-1}) = (\langle \psi_0|P|0^{m}\rangle, \dots, \langle \psi_{M-1}|P|0^{m}\rangle) \in S_{\mathbb{C}}^{M-1}.
\end{align*}
Defining the measures $\nu_1 = \tilde{\nu}_1 \circ \phi_1^{-1}$ and $\nu_2 = \tilde{\nu}_2 \circ \phi_2^{-1}$ on $S_{\mathbb{C}}^{M-1}$ as image measures, Lemma~\ref{lemma:Pauli sum} yields, for all $k = 0, 1, \dots, M-1$,
\begin{align*}
    \int_{S_{\mathbb{C}}^{M-1}} |z_k|^2 \, \nu_j(dz) = \frac{1}{M}, \quad j = 1, 2.
\end{align*}

Therefore, we obtain the following.
\begin{theorem}
    The $\nu_1$-QPE measure and the $\nu_2$-QPE measure are state-averaged.
\end{theorem}

\begin{remark}\label{remark:zero-padding case}
    In Section~\ref{subsec:qpe simulation}, zero padding yields $\ker A = \mathrm{span}\{|j\rangle \mid j = m_0, \dots, M-1\}$, which is orthogonal to $\mathrm{span}\{|j\rangle \mid j = 0, 1, \dots, m_0-1\}$.
    Moreover, the eigenvectors of the original problem are orthogonal to $\ker A$, so they are contained in $\mathrm{span}\{|j\rangle \mid j = 0, 1, \dots, m_0-1\}$.
    Therefore, in such cases, a state 1-design can be constructed using only $\{|j\rangle \mid j = 0, 1, \dots, m_0-1\}$, so the relevant averaging is over a $m_0$-dimensional subspace and yields equal weights $1/m_0$.
\end{remark}

\subsection*{The Fej\'er Kernel and its Properties}

For $n = 2, 3, \dots$, define $F_n \colon \mathbb{R} \to \mathbb{R}$ by
\begin{align*}
    F_n(x)
    =
    \begin{cases}
        n, & x = 0, \pm 2\pi, \pm 4\pi, \dots, \\[.5em]
        \displaystyle\frac{1}{n}\frac{\sin^2(nx/2)}{\sin^2(x/2)}, & \text{otherwise}.
    \end{cases}
\end{align*}
The function $F_n$ is called the Fej\'er kernel.

\begin{figure}[ht]
    \centering
    \begin{tikzpicture}
        \begin{axis}[
            width=13cm,
            height=8.5cm,
            xlabel={$x$},
            ylabel={$F_n(x)$},
            domain=-pi:pi,
            samples=801,
            smooth,
            thick,
            axis lines=middle,
            grid=none,
            xtick={-3.14159,-2.35619,-1.5708,-0.78540,0,0.78540,1.5708,2.35619,3.14159},
            xticklabels={$-\pi$,$-\frac{3\pi}{4}$,$-\frac{\pi}{2}$,$-\frac{\pi}{4}$,$0$,$\frac{\pi}{4}$,$\frac{\pi}{2}$,$\frac{3\pi}{4}$,$\pi$},
            ytick={8},
            yticklabels={$n$},
            ymin=-0.5,
            ymax=9,
            xmin=-3.5,
            xmax=3.5,
            tick label style={font=\small},
            label style={font=\normalsize},
        ]
        \addplot[blue!80!black, line width=1.2pt] {
            abs(x) < 0.001 ? 8 : (1.0/8.0)*(sin(deg(8*x/2))/sin(deg(x/2)))^2
        };
        \end{axis}
    \end{tikzpicture}
    \caption{Fej\'er Kernel ($n = 8$)}
\end{figure}
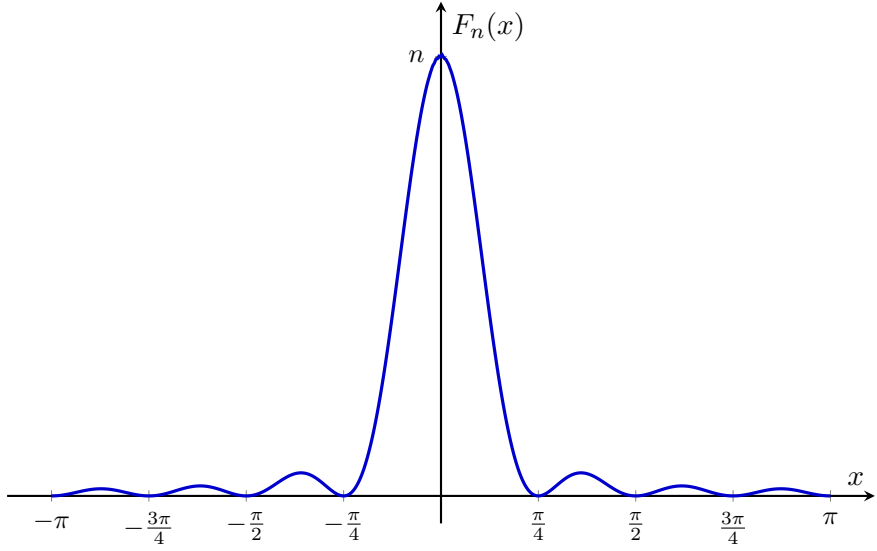

$F_n$ is a periodic even function with period $2\pi$, attaining its maximum value $n$ at $x = 0$.
For $k \in \mathbb{Z} \setminus n\mathbb{Z}$, $F_n(2k\pi/n) = 0$.

We state the following properties of the Fej\'er kernel.

\begin{lemma}\label{lemma:Fejer properties}
    Let $n\geq 4$. Then the following (i)--(iii) hold.
    \begin{itemize}
        \item[(i)] For every $x \in \mathbb{R}$,
            \begin{align*}
                \sum_{k=0}^{n-1} \frac{F_n(x + 2\pi k/n)}{F_n(0)} = 1.
            \end{align*}
        \item[(ii)] For each $k = 1, 2, \dots, \lfloor n/2 \rfloor - 1$,
            \begin{align*}
                \frac{F_n(x)}{F_n(0)}
                &\leq \frac{1}{k^2\pi^2} + \frac{1}{n^2}\left( 1 - \frac{4}{\pi^2} \right), \quad x\in I_k,
            \end{align*}
            where
            \begin{align*}
                I_k = \left[\frac{2k\pi}{n}, \frac{2(k+1)\pi}{n}\right],
            \end{align*}
            and $\lfloor \cdot \rfloor$ denotes the floor function.
        \item[(iii)] For every $x \in [0, \pi/n]$,
            \begin{align*}
                \frac{F_n(x)}{F_n(0)} \geq \frac{4}{\pi^2}.
            \end{align*}
    \end{itemize}
\end{lemma}
\begin{proof}
    (i) This is an immediate consequence of the well-known formula
    \begin{align*}
        F_n(x)
        = \sum_{j=-n+1}^{n-1} \left(1 - \frac{|j|}{n}\right) e^{ijx}.
    \end{align*}

    (ii) Let $x\in I_k$.
    Since $\sin^2(nx/2) \leq 1$ and $\sin^2(k\pi/n) \leq \sin^2(x/2)$, we have
    \begin{align*}
        \frac{F_n(x)}{F_n(0)}
        &= \frac{1}{n^2}\frac{\sin^2(nx/2)}{\sin^2(x/2)} \leq \frac{1}{n^2\sin^2(k\pi/n)}.
    \end{align*}

    The inequality
    \begin{align*}
        \frac{1}{\sin^2 x}
        \leq \frac{1}{x^2} + 1 - \frac{4}{\pi^2},
        \qquad 0 < x < \frac{\pi}{2},
    \end{align*}
    applied with $x = k\pi/n$, yields
    \begin{align*}
        \frac{1}{n^2\sin^2(k\pi/n)}
        \leq \frac{1}{k^2\pi^2} + \frac{1}{n^2}\left(1 - \frac{4}{\pi^2}\right).
    \end{align*}

    (iii) If $0 \le x \le \pi/n$, then
    \begin{align*}
        \frac{F_n(x)}{F_n(0)}
        =\frac{1}{n^2}\frac{\sin^2(nx/2)}{\sin^2(x/2)}
        \geq \frac{1}{n^2}\frac{\sin^2(nx/2)}{(x/2)^2}
        \geq \frac{4}{\pi^2},
    \end{align*}
    where we used $\sin t \leq t$ and $\sin t \geq 2t/\pi$ for $0 \leq t \leq \pi/2$.
\end{proof}

\subsection*{Peak Detection}
Let $\mu$ be a state-averaged QPE measure. That is,
\begin{align*}
    p_j 
    &:= \mu(\{a_j\} \times S_{\mathbb{C}}^{M-1})\\
    &= \frac{1}{M} \sum_{k=0}^{M-1} \tilde{F}_N(2\pi(\theta_k - a_j)), \quad j = 0, 1, \dots, N-1.
\end{align*}
Here, $\tilde{F}_N(x) = F_N(x)/F_N(0)$.

For $k = 0, 1, \dots, M-1$, define the closed interval on $\mathbb{T}$:
\begin{align*}
    I_k
    &= \left[\theta_k - \frac{1}{N},\, \theta_k + \frac{1}{N}\right]_{\mathbb{T}}.
\end{align*}
Set
\begin{align*}
    \tau &= \frac{4}{\pi^2}.
\end{align*}

\begin{theorem}[Peak detection]\label{theorem:peak detection}
    Assume that $N \geq 4M$ and
    \begin{align*}
        \frac{3}{N} \leq \min_{0 \leq k \leq M-1} |\theta_{k+1} - \theta_k|_{\mathbb{T}}.
    \end{align*}
    Define
    \begin{align*}
        C &= \left\{a_j \in A_N \,\middle|\, p_j \geq \frac{\tau}{M},~j = 0, 1, \dots, N-1\right\}.
    \end{align*}
    Then the following (1)--(5) hold.
    \begin{itemize}
        \item[(1)] If $k \neq j$, then $I_k \cap I_j = \emptyset$.
        \item[(2)] For all $k = 0, 1, \dots, M-1$, $C \cap I_k \neq \emptyset$.
        \item[(3)] $C \subset \bigcup_{k=0}^{M-1} I_k$.
        \item[(4)] $A_N \setminus \bigcup_{k=0}^{M-1} I_k \neq \emptyset$.
        \item[(5)] If $a_j \in C$, then $a_{j-1} \notin C$ or $a_{j+1} \notin C$.
    \end{itemize}
\end{theorem}

\begin{proof}
    (1) Suppose $k \neq j$ and $I_k \cap I_j \neq \emptyset$.
    Taking $x \in I_k \cap I_j$, we have $|\theta_k - \theta_j|_{\mathbb{T}} \leq |\theta_k - x|_{\mathbb{T}} + |\theta_j - x|_{\mathbb{T}} \leq 2/N$, contradicting $|\theta_k - \theta_j|_{\mathbb{T}} \geq 3/N$.

    (2) For $k = 0, 1, \dots, M-1$, the set $A_N \cap I_k$ contains at least two adjacent points.
    Take two adjacent points $a_j, a_{j+1} \in A_N \cap I_k$.
    Since either $|a_j-\theta_k|_{\mathbb{T}} \leq 1/N$ or $|a_{j+1}-\theta_k|_{\mathbb{T}} \leq 1/N$, by Lemma~\ref{lemma:Fejer properties}(iii), we have
    \begin{align*}
        \max\{p_j, p_{j+1}\} \geq \frac{\tau}{M}.
    \end{align*}
    Hence, $C \cap I_k \neq \emptyset$.

    (3) Let $a_j \in A_N \setminus \bigcup_{k=0}^{M-1} I_k$.
    Then, for every $k = 0,1,\dots,M-1$,
    \begin{align*}
        |\theta_k - a_j|_{\mathbb{T}} > \frac{1}{N}.
    \end{align*}

    Let $k'$ be such that $a_j \in (\theta_{k'}, \theta_{k'+1})_{\mathbb{T}}$.
    Then
    \begin{align}
        p_j
        &= \frac{1}{M} \sum_{k=0}^{M-1} \tilde{F}_N(2\pi(\theta_k-a_j)) \nonumber\\
        &= \frac{1}{M} \sum_{l=0}^{M-k'-2} \tilde{F}_N\bigl(2\pi(\theta_{k'+l+1}-a_j)\bigr)
        + \frac{1}{M} \sum_{l=0}^{k'} \tilde{F}_N\bigl(2\pi(\theta_{k'-l}-a_j)\bigr).
        \label{eq:estimate of p_j}
    \end{align}

    For each $l \geq 0$, the assumption
    \begin{align*}
        |\theta_{k+1}-\theta_k|_{\mathbb{T}} \geq \frac{3}{N},
        \qquad k=0,1,\dots,M-1,
    \end{align*}
    together with $a_j \in (\theta_{k'}, \theta_{k'+1})_{\mathbb{T}}$ implies that
    \begin{align*}
        |\theta_{k'+l+1}-a_j|_{\mathbb{T}} > \frac{3l+1}{N}, \quad |\theta_{k'-l}-a_j|_{\mathbb{T}} > \frac{3l+1}{N}.
    \end{align*}
    Therefore, by Lemma~\ref{lemma:Fejer properties}(ii),
    \begin{align}
        \tilde{F}_N\bigl(2\pi(\theta_{k'+l+1}-a_j)\bigr),
        \ \tilde{F}_N\bigl(2\pi(\theta_{k'-l}-a_j)\bigr)
        \leq \frac{1}{\pi^2(3l+1)^2}
        + \frac{1}{N^2}\left(1-\frac{4}{\pi^2}\right).
        \label{eq:F_N ineq}
    \end{align}

    By \eqref{eq:estimate of p_j} and \eqref{eq:F_N ineq}, it follows that
    \begin{align}
        p_j
        &\leq \frac{2}{\pi^2 M}\sum_{l=0}^{\infty}\frac{1}{(3l+1)^2}
        + \frac{1}{N^2}\left(1-\frac{4}{\pi^2}\right) \nonumber\\
        &\leq \frac{1}{M}\left\{
            \frac{2}{\pi^2}\sum_{l=0}^{\infty}\frac{1}{(3l+1)^2}
            + \frac{1}{32}\left(1-\frac{4}{\pi^2}\right)
        \right\}
        < \frac{\tau}{M},\label{eq:gap between C and not C}
    \end{align}
    where the second inequality follows from $N \geq 4M \geq 8$.

    Hence,
    \begin{align*}
        C \subset \bigcup_{k=0}^{M-1} I_k.
    \end{align*}

    (4) Assuming $A_N \subset \bigcup_{k=0}^{M-1} I_k$, since each $I_k \cap A_N$ has at most 3 elements, we get $N \leq 3M$, contradicting $N \geq 4M$.

    (5) Let $a_j \in C$.
    By (3), there exists $k=0,1,\dots,M-1$ such that $a_j \in I_k$.
    Then $|\theta_k-a_{j-1}|_{\mathbb{T}} \geq 1/N$ or $|\theta_k-a_{j+1}|_{\mathbb{T}} \geq 1/N$.
    Suppose that $|\theta_k-a_{j+1}|_{\mathbb{T}} \geq 1/N$.

    For $l \neq k$,
    \begin{align*}
        \frac{3}{N}
        &\leq |\theta_k-\theta_l|_{\mathbb{T}} \\
        &\leq |\theta_k-a_j|_{\mathbb{T}} + |a_j-a_{j+1}|_{\mathbb{T}} + |\theta_l-a_{j+1}|_{\mathbb{T}} \\
        &\leq \frac{2}{N} + |\theta_l-a_{j+1}|_{\mathbb{T}},
    \end{align*}
    and hence $|\theta_l-a_{j+1}|_{\mathbb{T}} \geq 1/N$.
    Therefore, $|\theta_l-a_{j+1}|_{\mathbb{T}} \geq 1/N$ for all $l = 0,1,\dots,M-1$.
    By the same argument as in the proof of~(3), we obtain
    \begin{align*}
        p_{j+1} < \frac{\tau}{M},
    \end{align*}
    and thus $a_{j+1} \notin C$.
\end{proof}

Theorem~\ref{theorem:peak detection} shows that, for a state-averaged QPE measure, $C$ contains the points relevant for detecting peaks within distance $1/N$ of each eigenphase.
In other words, an eigenphase exists in the neighborhood of every point contained in $C$.
The theorem also provides sufficient conditions for this to hold.

\begin{corollary}
    Assume that $N \geq 4M$ and
    \begin{align*}
        \frac{3}{N} \leq \min_{0 \leq k \leq M-1} |\theta_{k+1} - \theta_k|_{\mathbb{T}}.
    \end{align*}
    Then the following hold.
    \begin{itemize}
        \item[(1)] If $a_j, a_{j+1} \in C$, then $a_{j-1}, a_{j+2} \notin C$, and there exists $k =0, 1, \dots, M-1$ such that $\theta_k \in [a_j, a_{j+1}]_{\mathbb{T}}$.
        \item[(2)] If $a_j \in C$ and $a_{j-1}, a_{j+1} \notin C$, then there exists $k =0, 1, \dots, M-1$ such that 
        \begin{align*}
            \theta_k \in \left[a_j - \frac{1}{2N}, a_j + \frac{1}{2N}\right]_{\mathbb{T}}.
        \end{align*}
    \end{itemize}
\end{corollary}
\begin{proof}
    (1) It follows from Theorem~\ref{theorem:peak detection}(3) and (5).

    (2) Suppose $a_j \in C$ and $a_{j-1}, a_{j+1} \notin C$.
    Since $a_j\in C$, by Theorem~\ref{theorem:peak detection}(3), $a_j\in I_k$ for some $k$.
    If $\theta_k \notin [a_j - 1/(2N), a_j + 1/(2N)]_{\mathbb{T}}$, then $a_{j-1} \in I_k$ or $a_{j+1} \in I_k$.

    If $a_{j-1} \in I_k$, then $|\theta_k - a_{j-1}|_{\mathbb{T}} < 1/(2N)$, and by Lemma~\ref{lemma:Fejer properties}(iii), we get $a_{j-1} \in C$, contradicting $a_{j-1} \notin C$.
    The case $a_{j+1} \in I_k$ leads to a similar contradiction.
\end{proof}

\subsection*{Estimation of the Number of Samples}
Let $\{Y_k = (X_k, Z_k)\}_{k=1}^\infty$ be an i.i.d.\ sequence of $A_N \times S_{\mathbb{C}}^{M-1}$-valued random variables defined on a probability space $(\Omega, \mathcal{F}, \mathbb{P})$.
For $k = 1, 2, \dots$, denote the distribution of $Y_k$ by $\mathbb{P}^{Y_k}=\mathbb{P}\circ Y_k^{-1}$.
Assume that $\mathbb{P}^{Y_1}$ is a $\nu$-QPE measure, i.e.,
\begin{align*}
    \mathbb{P}(Y_1 \in \{a\} \times B)
    &= \int_B \nu(dz)\, \sum_{k=0}^{M-1} |z_k|^2 \, \tilde{F}_N(2\pi(\theta_k - a)), \quad a \in A_N,~B \in \mathcal{B}(S_{\mathbb{C}}^{M-1}).
\end{align*}
Writing $Z_k = (Z_k^0, Z_k^1, \dots, Z_k^{M-1})$, this means
\begin{align*}
    \mathbb{P}(X_1 = a \mid Z_1)
    &= \sum_{k=0}^{M-1} |Z_1^{k}|^2 \, \tilde{F}_N(2\pi(\theta_k - a)), \quad a \in A_N,
\end{align*}
where $\mathbb{P}(\cdot\mid Z_1)$ represents the conditional probability with respect to $Z_1$.

When $\mathbb{P}^{Y_1}$ is state-averaged,
\begin{align*}
    p_j := \mathbb{P}(X_1 = a_j) = \frac{1}{M}\sum_{k=0}^{M-1}\tilde{F}_N(2\pi(\theta_k - a_j)), \quad j = 0, 1, \dots, N-1.
\end{align*}
In what follows, we assume that $\mathbb{P}^{Y_1}$ is state-averaged.

For $K = 1, 2, \dots$, define
\begin{align*}
    \hat{p}_j
    &= \hat{p}_j(K) = \frac{1}{K}\sum_{k=1}^K \boldsymbol{1}_{\{a_j\}}(X_k), \quad j = 0, 1, \dots, N-1,
\end{align*}
where $\boldsymbol{1}_{\{a_j\}}$ is the indicator function of $\{a_j\}$.
By the strong law of large numbers, for each $j = 0, 1, \dots, N-1$, $\hat{p}_j$ converges almost surely to $p_j$ as $K \to \infty$.

Define
\begin{align*}
    & \sigma = \frac{2}{\pi^2}\sum_{l=0}^{\infty} \frac{1}{(3l+1)^2}, \quad d_N = \frac{1}{N^2}\left( 1 - \frac{4}{\pi^2} \right).
\end{align*}
Note that
\begin{align*}
    \sigma < \tau, \quad \frac{\sigma}{M} + d_N < \frac{\tau}{M}
\end{align*}
by \eqref{eq:gap between C and not C}.

Recall that by definition, if $a_j \in C$ then $p_j \geq \tau/M$.
On the other hand, if $a_j \in D := A_N \setminus \bigcup_{k=0}^{M-1} I_k$, then $a_j \notin C$ and by~\eqref{eq:gap between C and not C}, $p_j \leq \sigma/M + d_N$.
Denote by $I_C$ the set of indices $j$ with $a_j \in C$, and by $I_D$ the set of indices $j$ with $a_j \in D$.

\begin{lemma}\label{lemma:UB of p}
    Assume that $N \geq 4M$ and
    \begin{align*}
        \frac{3}{N} \leq \min_{0 \leq k \leq M-1} |\theta_{k+1} - \theta_k|_{\mathbb{T}}.
    \end{align*}
    Then, for all $j \in I_C$,
    \begin{align*}
        p_j \leq \frac{\gamma}{M} + d_N,
    \end{align*}
    where
    \begin{align*}
        \gamma &= 1 + \frac{1}{\pi^2}\left( \frac{\pi^2}{6} - \sum_{l=0}^\infty \frac{1}{(3l+1)^2} \right).
    \end{align*}
\end{lemma}

\begin{proof}
    Let $j \in I_C$. Then $a_j \in C$.
    By Theorem~\ref{theorem:peak detection}(3), we can take $k'$ such that $a_j \in I_{k'}$.

    Let $l \geq 1$. 
    If $a_j \in [\theta_{k'} - 1/N, \theta_{k'}]_{\mathbb{T}}$, then
    \begin{align*}
        |\theta_{k' + l} - a_j|_{\mathbb{T}} \geq \frac{3l}{N}, 
        \quad
        |\theta_{k' - l} - a_j|_{\mathbb{T}} \geq \frac{3l - 1}{N}.
    \end{align*}
    If $a_j \in [\theta_{k'}, \theta_{k'} + 1/N]_{\mathbb{T}}$, then
    \begin{align*}
        |\theta_{k' + l} - a_j|_{\mathbb{T}} \geq \frac{3l - 1}{N}, 
        \quad
        |\theta_{k' - l} - a_j|_{\mathbb{T}} \geq \frac{3l}{N}.
    \end{align*}

    By Lemma~\ref{lemma:Fejer properties}(ii), and using the bound
    $\tilde{F}_N(2\pi(\theta_{k'} - a_j)) \leq 1$, we obtain
    \begin{align*}
        p_j
        &\leq \frac{1}{M} + \frac{1}{M} \sum_{\substack{0\leq k\leq M-1\\k \neq k'}} \tilde{F}_N(2\pi(\theta_k - a_j)) \\
        &\leq \frac{1}{M} + \frac{1}{M} \sum_{l \equiv 0,2\!\!\!\pmod{3}} \frac{1}{\pi^2 l^2}
            + \frac{1}{N^2}\left( 1 - \frac{4}{\pi^2} \right) \\
        &= \frac{1}{M} \left\{ 1 + \frac{1}{\pi^2}\left( \frac{\pi^2}{6} - \sum_{l=0}^\infty \frac{1}{(3l+1)^2} \right) \right\} + \frac{1}{N^2}\left( 1 - \frac{4}{\pi^2} \right),
    \end{align*}
    where we used
    \begin{align*}
        \sum_{l=1}^\infty \frac{1}{l^2} &= \frac{\pi^2}{6}. \qedhere
    \end{align*}
\end{proof}

In QPE, $C$ provides the approximate peak locations near eigenphases, and confusing a point in $I_C$ with a point in $I_D$ would result in an incorrect eigenphase estimate.
The following provides a probability bound on such confusion and the shot count estimate needed to keep this probability below a prescribed level.

\begin{theorem}[Sample-complexity estimate]\label{theorem:shot estimate}
    Assume that $M\geq 3$, $N \geq 4M$, and
    \begin{align*}
        \frac{3}{N} \leq \min_{0 \leq k \leq M-1} |\theta_{k+1} - \theta_k|_{\mathbb{T}}.
    \end{align*}
    Let $\varepsilon>0$ satisfy
    \begin{align*}
        0 < \varepsilon \leq \frac{\tau - \sigma}{2} - \frac{Md_N}{2}.
    \end{align*}
    Then
    \begin{align*}
        &\mathbb{P}\left(\left\{\min_{j \in I_C} \hat{p}_j \leq \frac{\tau}{M} - \frac{\varepsilon}{M}\right\} \cup \left\{ \max_{j \in I_D} \hat{p}_j \geq \frac{\sigma}{M} + \frac{\varepsilon}{M} + d_N \right\}\right) \\[.5em]
        & \hspace{3em} \leq 2M \exp\left\{-KH^{-}\left(\frac{\gamma}{M}+d_N, \frac{\varepsilon}{M}\right)\right\} + (N-M) \exp\left\{-KH^{+}\left(\frac{\gamma}{M}+d_N, \frac{\varepsilon}{M}\right)\right\} \\[.5em]
        & \hspace{3em} \leq (N+M) \exp\left\{-KH^{+}\left(\frac{\gamma}{M}+d_N, \frac{\varepsilon}{M}\right)\right\},
    \end{align*}
    where, for $0 < a < 1/3$,
    \begin{align*}
        & H^{+}(x, a) = (x + a)\log\frac{x + a}{x} + (1 - x - a)\log\frac{1 - x - a}{1 - x}, \quad 0<x<1-a,\\
        & H^{-}(x, a) = (x - a)\log\frac{x - a}{x} + (1 - x + a)\log\frac{1 - x + a}{1 - x}, \quad a<x<1.
    \end{align*}

    Furthermore, for $0 < \delta < 1$, if
    \begin{align*}
        K
        &\geq \left\{H^{+}\left(\frac{\gamma}{M}+d_N, \frac{\varepsilon}{M}\right)\right\}^{-1}\log\frac{N+M}{\delta} \\[.5em]
        & \sim \left\{(\gamma+\varepsilon)\log\frac{\gamma+\varepsilon}{\gamma} - \varepsilon \right\}^{-1}\,M\log\frac{N+M}{\delta}, \quad M \to \infty,
    \end{align*}
    then the above probability is at most $\delta$.
\end{theorem}
\begin{proof}
    Let $j \in I_C$.
    Since $p_j \geq \tau/M$, we have 
    \begin{align*}
        \left\{ \omega\in\Omega\,\middle|\, \hat{p}_j(\omega) \leq \frac{\tau}{M}-\frac{\varepsilon}{M} \right\} \subset \left\{ \omega\in\Omega \,\middle|\, p_j - \hat{p}_j(\omega) \geq \frac{\varepsilon}{M} \right\}.
    \end{align*}
    For $\lambda \geq 0$, using the independence of $\{X_k\}_{k=1}^\infty$,
    \begin{align*}
        \mathbb{P}\left(\hat{p}_j \leq \frac{\tau}{M} - \frac{\varepsilon}{M}\right)
        &\leq \mathbb{P}\left(p_j - \hat{p}_j \geq \frac{\varepsilon}{M}\right) \\[.5em]
        &\leq \mathbb{P}\left(e^{\lambda(p_j - \hat{p}_j)} \geq e^{\frac{\lambda\varepsilon}{M}}\right) \\[.5em]
        &\leq e^{-\lambda\varepsilon/M} \mathbb{E}[e^{\lambda(p_j - \hat{p}_j)}] \\[.5em]
        &= e^{\varphi(\lambda)},
    \end{align*}
    where $\varphi(\lambda) = -\lambda(\varepsilon/M - p_j) + K\log(e^{-\lambda/K}p_j + 1 - p_j)$.

    Let
    \begin{align*}
        \lambda_0 &= K\log\frac{p_j(1 - p_j + \varepsilon/M)}{(p_j - \varepsilon/M)(1 - p_j)},
    \end{align*}
    which minimizes $\varphi(\lambda)$ over $\lambda \geq 0$, that is,
    \begin{align*}
        \inf_{\lambda \geq 0} \varphi(\lambda)
        &= \varphi(\lambda_0) = -KH^{-}\left(p_j, \frac{\varepsilon}{M}\right).
    \end{align*}
    Taking $\lambda = \lambda_0$, we obtain
    \begin{align*}
        \mathbb{P}\left(\hat{p}_j \leq \frac{\tau}{M} - \frac{\varepsilon}{M}\right)
        &\leq \exp\left\{-KH^{-}\left(p_j, \frac{\varepsilon}{M}\right)\right\} \\[.5em]
        &\leq \exp\left\{-KH^{-}\left(\frac{\gamma}{M}+d_N, \frac{\varepsilon}{M}\right)\right\},
    \end{align*}
    where we used $p_j \leq \gamma/M+d_N$ and the fact that $x\mapsto H^{-}(x, \varepsilon/M)$ is decreasing on $(\varepsilon/M, \gamma/M+d_N]$ (Remark \ref{remark: KL divergence}).

    By Theorem~\ref{theorem:peak detection}(5), for each $k = 0, 1, \dots, M-1$, the number of indices $j$ with $a_j \in C \cap I_k$ is at most two.
    Combining this with Theorem~\ref{theorem:peak detection}(2), $|I_C| \leq 2M$, and hence
    \begin{align}
        \mathbb{P}\left(\min_{j \in I_C} \hat{p}_j \leq \frac{\tau}{M} - \frac{\varepsilon}{M}\right)
        &= \mathbb{P}\left(\bigcup_{j \in I_C} \left\{\hat{p}_j \leq \frac{\tau}{M} - \frac{\varepsilon}{M}\right\}\right) \nonumber\\[.5em]
        & \leq \sum_{j \in I_C} \mathbb{P}\left(\hat{p}_j \leq \frac{\tau}{M} - \frac{\varepsilon}{M}\right) \nonumber\\
        &\leq 2M \exp\left\{-KH^{-}\left(\frac{\gamma}{M}+d_N, \frac{\varepsilon}{M}\right)\right\}. \label{eq:case j in I_C}
    \end{align}

    Let $j \in I_D$.
    In the same manner as for $j \in I_C$, we obtain
    \begin{align}
        \mathbb{P}\left(\hat{p}_j \geq \frac{\sigma}{M} + \frac{\varepsilon}{M} + d_N\right)
        &\leq \mathbb{P}\left(\hat{p}_j - p_j \geq \frac{\varepsilon}{M}\right) \nonumber\\[.5em]
        &\leq \exp\left\{-KH^{+}\left(\frac{\gamma}{M}+d_N, \frac{\varepsilon}{M}\right)\right\}. \label{eq:case j in I_D}
    \end{align}

    By Theorem~\ref{theorem:peak detection}(1) and (2), $C$ contains at least $M$ elements, and thus $|I_D| \leq N - M$.
    Hence we have
    \begin{align*}
        \mathbb{P}\left(\max_{j \in I_D} \hat{p}_j \geq \frac{\sigma}{M} + \frac{\varepsilon}{M} + d_N \right)
        &\leq (N - M) \exp\left\{-KH^{+}\left(\frac{\gamma}{M}+d_N, \frac{\varepsilon}{M}\right)\right\}.
    \end{align*}

    The final inequality follows from
    \begin{align*}
        H^{-}\left(\frac{\gamma}{M}+d_N, \frac{\varepsilon}{M}\right)
        &\geq H^{+}\left(\frac{\gamma}{M}+d_N, \frac{\varepsilon}{M}\right).\qedhere
    \end{align*}
\end{proof}

\begin{remark}\label{remark: KL divergence}
    In the proof of Theorem~\ref{theorem:shot estimate}, we used the following facts.
    \begin{itemize}
        \item[(1)] There is an $x_+\in ((1-a)/2,1/2)$ such that $x\mapsto H^{+}(x,a)$ is decreasing on $(a,x_+]$.
        \item[(2)] There is an $x_-\in (1/2,(1+a)/2)$ such that $x\mapsto H^{-}(x,a)$ is decreasing on $(0,x_-]$.
        \item[(3)] $H^{-}(x,a)\geq H^{+}(x,a)$ on $(a,1/2]$.
    \end{itemize}
    In \eqref{eq:case j in I_D}, (1) was used.
    Letting $a=\varepsilon/M$, we have
    \begin{align*}
        \frac{\gamma}{M} + d_N \leq \frac{1-a}{2} \quad \Longleftrightarrow \quad M\geq 2\gamma + \varepsilon + 2Md_N.
    \end{align*}
    Since
    \begin{align*}
        Md_N 
        &= \frac{M(1-\tau)}{N^2} \leq \frac{1-\tau}{16}
    \end{align*}
    and
    \begin{align*}
        2\gamma + \varepsilon + 2Md_N
        &= 2\gamma + \varepsilon + \frac{Md_N}{2} + \frac{3Md_N}{2} \\
        &\leq 2\gamma + \frac{\tau-\sigma}{2} + \frac{3}{2}\frac{1-\tau}{16} = 2.25\dots,
    \end{align*}
    (1) is applicable if $M\geq 3$.

    Moreover, the same bound implies
    \begin{align*}
        \frac{\gamma}{M}+d_N \leq \frac{1-a}{2} \leq \frac12,
    \end{align*}
    so (2) and (3) are also applicable at the point $x=\gamma/M+d_N$ used in the proof.
\end{remark}

\subsection*{Peak Detection from Samples}
Let $\varepsilon>0$ satisfy
\begin{align*}
    0 < \varepsilon \leq \frac{\tau - \sigma}{2} - \frac{Md_N}{2}.
\end{align*}
For each $\omega\in \Omega$, define
\begin{align*}
    \hat{C}(\omega) &= \left\{ a_j\in A_N \,\middle|\, \hat{p}_j(\omega)\geq \frac{\tau}{M}-\frac{\varepsilon}{M} \right\},
\end{align*}
and define
\begin{align*}
    \hat{\Omega}
    &= \left\{ \omega\in\Omega \,\middle|\, \min_{j\in I_C}\hat{p}_j(\omega) \geq \frac{\tau}{M}-\frac{\varepsilon}{M} \right\} \cap \left\{ \omega\in\Omega \,\middle|\, \max_{j\in I_D}\hat{p}_j(\omega) < \frac{\sigma}{M}+\frac{\varepsilon}{M}+d_N \right\}.
\end{align*}
Then, since 
\begin{align*}
    \frac{\sigma}{M} + \frac{\varepsilon}{M} + d_N \leq \frac{\tau}{M} - \frac{\varepsilon}{M},
\end{align*}
it is immediate that
\begin{align}
    \hat{\Omega}
    &\subset \left\{ \omega\in\Omega \,\middle|\, C\subset \hat{C}(\omega)\subset \bigcup_{k=0}^{M-1}I_k \right\}. \label{eq:inclusion relation in C^}
\end{align}

\begin{theorem}\label{theorem:peak detection from samples}
    Assume that $N \geq 4M$ and
    \begin{align*}
        \frac{3}{N} < \min_{0 \leq k \leq M-1} |\theta_{k+1} - \theta_k|_{\mathbb{T}}.
    \end{align*}
    Let $\omega\in\hat{\Omega}$.
    Then the following hold.
    \begin{itemize}
        \item[(1)] For all $k=0,1,\dots, M-1$, $\hat{C}(\omega)\cap I_k\neq \emptyset$.
        \item[(2)] If $a_{j-1},a_{j},a_{j+1}\in\hat{C}(\omega)$, then there exists $k=0,1,\dots,M-1$ such that $\theta_k=a_j$.
        Moreover, $a_{j\pm 2}\notin \hat{C}(\omega)$.
        \item[(3)] If $a_j, a_{j+1} \in \hat{C}(\omega)$ and $a_{j-1}, a_{j+2} \notin \hat{C}(\omega)$, then there exists $k =0, 1, \dots, M-1$ such that $\theta_k \in [a_j, a_{j+1}]_{\mathbb{T}}$.
        \item[(4)] If $a_j \in \hat{C}(\omega)$ and $a_{j-1}, a_{j+1} \notin \hat{C}(\omega)$, then there exists $k =0, 1, \dots, M-1$ such that
        \begin{align*}
            \theta_k \in \left[a_j - \frac{1}{2N}, a_j + \frac{1}{2N}\right]_{\mathbb{T}}.
        \end{align*}
    \end{itemize}
\end{theorem}
\begin{proof}
    (1) It follows immediately from Theorem~\ref{theorem:peak detection}(2) and \eqref{eq:inclusion relation in C^}.

    (2) Let $a_{j-1},a_{j},a_{j+1}\in\hat{C}(\omega)$.
    By \eqref{eq:inclusion relation in C^}, we can take $k$ such that $a_{j-1}, a_j, a_{j+1} \in I_k$.
    Then, since $I_k$ contains exactly three adjacent points of $A_N$, $\theta_k=a_j$.

    As in the proof of Theorem~\ref{theorem:peak detection}(5), we obtain
    \begin{align*}
        |\theta_l - a_{j\pm 2}|_{\mathbb{T}}>\frac{1}{N}, \quad l=0,1,\dots,M-1.
    \end{align*}
    Hence $a_{j\pm 2}\notin \bigcup_{k=0}^{M-1}I_k$, and therefore, by \eqref{eq:inclusion relation in C^}, $a_{j\pm 2}\notin \hat{C}(\omega)$.

    (3) Let $a_j, a_{j+1} \in \hat{C}(\omega)$.
    By \eqref{eq:inclusion relation in C^}, $a_j,a_{j+1}\in I_k$ for some $k$.
    Now suppose that $\theta_k\notin [a_j,a_{j+1}]_{\mathbb{T}}$.
    Then
    \begin{align*}
        |\theta_k-a_j|_{\mathbb{T}}>\frac{1}{N}
        \quad \text{or} \quad
        |\theta_k-a_{j+1}|_{\mathbb{T}}>\frac{1}{N}.
    \end{align*}
    
    If $|\theta_k-a_j|_{\mathbb{T}}>1/N$, then, by the same argument as in (2), we conclude that $a_j\notin \hat{C}(\omega)$.
    Similarly, if $|\theta_k-a_{j+1}|_{\mathbb{T}}>1/N$, then $a_{j+1}\notin \hat{C}(\omega)$.
    This contradicts the assumption.

    (4) Suppose $a_j \in \hat{C}(\omega)$ and $a_{j-1}, a_{j+1} \notin \hat{C}(\omega)$.
    Since $a_j \in \hat{C}(\omega)$, by \eqref{eq:inclusion relation in C^}, $a_j\in I_k$ for some $k$.
    If $\theta_k \notin [a_j - 1/(2N), a_j + 1/(2N)]_{\mathbb{T}}$, then $a_{j-1} \in I_k$ or $a_{j+1} \in I_k$.    

    If $a_{j-1} \in I_k$, then $|\theta_k - a_{j-1}|_{\mathbb{T}} < 1/(2N)$.
    By Lemma~\ref{lemma:Fejer properties}(iii), we have
    \begin{align*}
        p_{j-1} \geq \frac{\tau}{M}.
    \end{align*}
    Hence $a_{j-1} \in C$.
    Since $\omega \in \hat{\Omega}$ and $C \subset \hat{C}(\omega)$ by \eqref{eq:inclusion relation in C^}, this contradicts $a_{j-1} \notin \hat{C}(\omega)$.
    The case $a_{j+1} \in I_k$ leads to a similar contradiction.
\end{proof}

By \eqref{eq:inclusion relation in C^} and Theorem~\ref{theorem:peak detection from samples}, for every $\omega \in \hat{\Omega}$,
\begin{align*}
    C \subset \hat{C}(\omega) \subset \bigcup_{k=0}^{M-1} I_k,
\end{align*}
and properties~(1)--(4) hold.

By Theorem~\ref{theorem:shot estimate}, if $M \geq 3$, then for sufficiently large $K$,
\begin{align*}
    \mathbb{P}(\hat{\Omega}) \geq 1-\delta.
\end{align*}
Hence, with probability at least $1-\delta$, the above properties hold.

\section*{Acknowledgements}

The authors acknowledge the use of OpenAI ChatGPT 5.4 and Google Gemini 3.1 for editorial assistance, including suggestions on English phrasing, clarity of exposition, and manuscript organization. 
Google Gemini 3.1 was also used to assist in drafting code for the numerical experiments; the authors reviewed and validated the implementation and all reported results.



\begin{thebibliography}{99}

    \bibitem{bai2026random}
    Y.\,Bai, F.\,Xiong, and X.\,Kuang,
    Random State Approach to Quantum Computation of Electronic-Structure Properties,
    \textit{Chinese Physics Letters} \textbf{43} (2026), 010604.
    DOI:10.1088/0256-307X/43/1/010604.

    \bibitem{dutkiewicz2022heisenberg}
    A.\,Dutkiewicz, B.\,M.\,Terhal, and T.\,E.\,O'Brien,
    Heisenberg-Limited Quantum Phase Estimation of Multiple Eigenvalues with Few Control Qubits,
    \textit{Quantum} \textbf{6} (2022), 830.
    DOI:10.22331/q-2022-10-06-830.

    \bibitem{GSLW19}
    A.\,Gily\'{e}n, Y.\,Su, G.\,H.\,Low, and N.\,Wiebe,
    Quantum Singular Value Transformation and Beyond: Exponential Improvements for Quantum Matrix Arithmetics,
    \textit{Proceedings of the 51st Annual ACM SIGACT Symposium on Theory of Computing (STOC 2019)},
    2019.
    arXiv:1806.01838 [quant-ph].

    \bibitem{harrow2009quantum}
    A.\,W.\,Harrow, A.\,Hassidim, and S.\,Lloyd,
    Quantum Algorithm for Linear Systems of Equations,
    \textit{Physical Review Letters} \textbf{103} (2009), 150502.
    DOI:10.1103/PhysRevLett.103.150502.

    \bibitem{goh2026dos}
    M.\,L.\,Goh and B.\,Koczor,
    Direct Estimation of the Density of States for Fermionic Systems,
    \textit{Physical Review Research} \textbf{8} (2026), 013250.
    DOI:10.1103/h8l5-87zt.

    \bibitem{lim2024curvefitted}
    S.\,M.\,Lim, C.\,E.\,Susa, and R.\,Cohen,
    Curve-Fitted QPE: Extending Quantum Phase Estimation Results for a Higher Precision Using Classical Post-Processing,
    arXiv:2409.15752 [quant-ph], 2024.

    \bibitem{low2019hamiltonian}
    G.\,H.\,Low and I.\,L.\,Chuang,
    Hamiltonian Simulation by Qubitization,
    \textit{Quantum} \textbf{3} (2019), 163.
    arXiv:1610.06546 [quant-ph].

    \bibitem{martyn2021grand}
    J.\,M.\,Martyn, Z.\,M.\,Rossi, A.\,K.\,Tan, and I.\,L.\,Chuang,
    A Grand Unification of Quantum Algorithms,
    \textit{PRX Quantum} \textbf{2} (2021), 040203.
    DOI:10.1103/PRXQuantum.2.040203.

    \bibitem{mele2024haar}
    A.\,A.\,Mele,
    Introduction to Haar Measure Tools in Quantum Information: A Beginner's Tutorial,
    \textit{Quantum} \textbf{8} (2024), 1340.
    DOI:10.22331/q-2024-05-08-1340.

    \bibitem{nielsen2010quantum}
    M.\,A.\,Nielsen and I.\,L.\,Chuang,
    \textit{Quantum Computation and Quantum Information},
    Cambridge University Press, Cambridge, 2010.

    \bibitem{obrien2019quantum}
    T.\,E.\,O'Brien, B.\,Tarasinski, and B.\,M.\,Terhal,
    Quantum Phase Estimation of Multiple Eigenvalues for Small-Scale (Noisy) Experiments,
    \textit{New Journal of Physics} \textbf{21} (2019), 023022.

    \bibitem{shor1994algorithms}
    P.\,W.\,Shor,
    Algorithms for Quantum Computation: Discrete Logarithms and Factoring,
    \textit{Proceedings of the 35th Annual Symposium on Foundations of Computer Science (FOCS 1994)},
    IEEE, 1994, pp.\,124--134.
    DOI:10.1109/SFCS.1994.365700.

    \bibitem{somma2019eigenvalue}
    R.\,D.\,Somma,
    Quantum Eigenvalue Estimation via Time Series Analysis,
    \textit{New Journal of Physics} \textbf{21} (2019), 123025.
    DOI:10.1088/1367-2630/ab5c60.

    \bibitem{wilde2017quantum}
    M.\,M.\,Wilde,
    \textit{Quantum Information Theory},
    Cambridge University Press, Cambridge, 2017.

\end{thebibliography}
\end{document}